\documentclass[12pt,reqno,oneside]{amsart}

\usepackage{amsthm}
\usepackage{amssymb}
\usepackage{mathrsfs}
\usepackage{tikz}
\usetikzlibrary{arrows, 
	positioning, 
	calc, 
	decorations.pathreplacing, 
	matrix, 
	patterns, 
	patterns.meta, 
	fadings, 
	math, 
	fit
}
\usepackage[margin=1in]{geometry}
\usepackage{xcolor}
\usepackage{enumitem}
\usepackage{bbm}
\usepackage[colorlinks,
linkcolor=blue,
anchorcolor=green,
citecolor=blue,
]{hyperref}

 \usepackage{rotating}

\newcommand{\bbC}{{\mathbb{C}}}
\newcommand{\bbD}{{\mathbb{D}}}

\newcommand{\bbQ}{{\mathbb{Q}}}
\newcommand{\T}{{\mathbb{T}}}
\newcommand{\bbR}{{\mathbb{R}}}

\newcommand{\bbZ}{{\mathbb{Z}}}

\newcommand{\SL}{{\mathrm{SL}}}
\newcommand{\SU}{{\mathrm{SU}}}

\newcommand{\tr}{{\mathrm{Tr}}}

\renewcommand{\Im}{\operatorname{Im}}
\renewcommand{\Re}{\operatorname{Re}}

\allowdisplaybreaks

\newtheorem{theorem}{Theorem}[section]

\newtheorem{lemma}[theorem]{Lemma}

\theoremstyle{definition}
\newtheorem{definition}[theorem]{Definition}
\newtheorem{remark}[theorem]{Remark}

\definecolor{purple}{rgb}{.5,0,1}
\definecolor{orange}{rgb}{1,.5,0}
\definecolor{green}{rgb}{0,.4,0}

\numberwithin{equation}{section}

\makeatletter
 
\newcommand{\Rmnum}[1]{\expandafter\@slowromancap\romannumeral #1@} 
\makeatother

\title[Spectral Transitions]{Spectral Transitions and Singular Continuous Spectrum in a New Family of Quasi-Periodic Quantum Walks}
\author[X.\ Yang]{Xinyu Yang}
\email{\href{meiyang010420@163.com}{meiyang010420@163.com}}
\address{[X.\ Yang] Chern Institute of Mathematics and LPMC, Nankai University, Tianjin 300071, China}

\author[L.\ Li]{Long Li}
\email{\href{mailto:ll106@rice.edu}{longli@tamu.edu}}
\address{[L.\ Li] Department of Mathematics, Texas A\&M University, College Station, TX 77843, USA}

\author[Q.\ Zhou]{Qi Zhou}
\email{\href{qizhou@nankai.edu.cn}{qizhou@nankai.edu.cn}}
\address{[Q.\ Zhou] Chern Institute of Mathematics and LPMC, Nankai University, Tianjin 300071, China}

\begin{document}

\maketitle
\begin{center}
\emph{Dedicated to Barry Simon on the occasion of his 80th birthday.}
\end{center}

\vspace{1em}

\begin{abstract}
This paper introduces and rigorously analyzes a new class of one-dimensional discrete-time quantum walks whose dynamics are governed by a parametrized family of extended CMV matrices. The model generalizes the unitary almost Mathieu operator (UAMO) and exhibits a richer spectral phase diagram, closely resembling the extended Harper’s model.  It provides the first example of a solvable quasi-periodic quantum walk that exhibits a stable region of purely singular continuous spectrum.  

\end{abstract}

\setcounter{tocdepth}{1}

\tableofcontents

\section{Introduction}
\subsection{Quantum Walks and Quantum Search Algorithm}Quantum walks are quantum analogues of classical random walks, in which the evolution is governed by unitary operators and the probability amplitudes can interfere. This interference leads to behaviors absent in classical walks, such as ballistic transport and localization phenomena, and underlies many quantum algorithms. A prominent example is Grover's search algorithm \cite{Grover1996}, which can be interpreted as a quantum walk on a complete graph. The amplitude of a marked state is amplified through repeated applications of a unitary operator, performing amplitude amplification \cite{Brassard2002} within a 2D subspace. This illustrates a key principle: structured unitary evolution can concentrate amplitude onto target subspaces, yielding quadratic speedup over classical search, compare \cite{Shenvi2003,Childs2004, Kempe2003}.

Discrete-time quantum walks are typically described by the composition of a parametrized shift operator \(S\) acting on position space and a coin operator \(Q \in U(2)\) acting on the internal degrees of freedom. The evolution operator $W = S \cdot Q$ governs the dynamics of the system. The interplay between \(S\) and \(Q\) produces interference patterns that can enhance or suppress amplitudes at different sites, analogous to amplitude amplification in Grover-type algorithms. In particular, the spectral properties of \(W\) play a crucial role in determining the transport and localization behavior of the quantum walker.

In this work we study a special class of one-dimensional, discrete-time quantum walks in which \(S\) is a parametrized shift and \(Q\in U(2)\) is a general coin operator. We give a rigorous spectral analysis of the walk operator \(W=S\cdot Q\) and prove results characterizing the three possible spectral types: pure point spectrum (Anderson localization), purely absolutely continuous spectrum, and singular continuous spectrum. By the RAGE theorem (cf. \cite{AmreinGeorgescu1996, Ruelle1969, AmreinGeorgescu1973, Enss1978}), these spectral types determine the long-time dynamical behavior: absolutely continuous spectrum is associated with ballistic spreading and rapid exploration of the lattice, while pure point spectrum yields exponential (or strong) localization with amplitudes trapped near a finite set of sites.

Singular continuous spectrum is a subtle regime between localization and ballistic transport. Dynamically, it gives rise to anomalous diffusion characterized by slow, irregular spreading and nontrivial transport exponents. It reflects persistent coherence in physics, mathematically it often accompanies structurally complex spectral measures, often of a fractal or Cantor type. In quantum walks, this leads to an intermediate behavior where amplitude transfer is too irregular for efficient search, yet not fully localized, resulting in parameter-sensitive, incoherent transport over long times.

The spectral decomposition of 
\(S \cdot Q\) thus serves both to classify long-term behavior and to guide practical design. For applications requiring fast, reliable transport, regimes dominated by singular continuous spectrum should be avoided; conversely, such regimes can be harnessed to sustain amplitude or engineer persistent states, though their dynamics tend to be less stable and more sensitive than those in purely absolutely continuous or pure point regimes.

Remarkably, for a specific and natural choice of the shift and coin, this unitary operator finds an exact representation in the form of a CMV matrix.
 The CMV matrices named after M. J. Cantero, L. Moral, L. Vel\'azquez arose from the study of minimal matrix representations of unitary operators \cite{CMV2003} and have a natural connection with the theory of orthogonal polynomials on the unit circle (OPUC). For a detailed background including a comprehensive list of developments and open problems, we refer interested readers to Barry Simon's two monographs \cite{OPUC1, OPUC2}. These matrices are five-diagonal unitary matrices representing the multiplication operator in $L^2(\partial\mathbb{D},d\mu)$ under a suitable basis, where $\mu$ is the orthogonal probability measure supported on the unit circle. The canonical spectral measure of the CMV matrix associated to a given sequence of  Verblunsky coefficients that describe the three-term iterative relation of the orthogonal polynomials turns out to be the orthogonal probability measure on the unit circle. Moreover, by the Verblunky's Theorem (unitary analogue of the Favard's Theorem for orthogonal polynomials on the real line (OPRL)), the correspondence between the set of Verblunsky coefficients and the set of non-trivial probability measures is bijective. Indeed, the CMV matrices and their bi-infinite extensions are related to the OPUC in the same way as Jacobi matrices related to the OPRL.
 
 One of the major topics in spectral theory is to study the properties of the measure when certain information of the coefficients is detected. The most well-known example is perhaps the almost Mathieu operator (AMO), defined by
 \begin{equation}\label{amo}
(H_{\lambda,\alpha,x}u)_n = u_{n+1} + u_{n-1} + 2\lambda \cos(2\pi(n\alpha +x))u_n.
\end{equation}
Intensive activities over the past decades give rather comprehensive understanding of its spectral properties, see \cite{CEY1990,Jitomirskaya1999, Puig2004CMP, AD08, AJ09, AJ10, AJM2017, JL2018,AYZ2017,AJZ2018, AYZ2023, JL2024} and the references therein. 
 A unitary analogue of the AMO  was formulated by \cite{FOZ17} and studied thoroughly by \cite{CFO23}. Various spectral properties that have been established in \cite{FOZ17, CFO23, CFLOZ24, CF2024, CF2024a, Yang2024} are parallel to the case of AMO, justifying the name of ``unitary almost Mathieu operator" dubbed by \cite{FOZ17}. What's interesting is quite recently,  a physical realization by \cite{LCZX2025} verifies the metal-insulator transition of a discrete time quantum walk modeled by the UAMO. 

A well-known counterpart in the spectral theory of quasi-periodic operators is the extended Harper’s model (EHM), introduced by Thouless in 1983. This model is defined on  $\ell^2(\bbZ)$ by the quasi-periodic Jacobi operator:
\[(H_{\lambda,\theta,\Phi}u)(n)=v(\theta+n\Phi)u(n)+c_\lambda(\theta+n\Phi)u(n+1)+\overline{c_{\lambda}(\theta+(n-1)\Phi)}u(n-1)\]
with \[c_\lambda(\theta)=\lambda_1 e^{-2\pi i(\theta+\frac{\Phi}{2})}+\lambda_2+\lambda_3e^{2\pi i(\theta+\frac{\Phi}{2})},~v(\theta)=2\cos(2\pi\theta)\]
where we used $\lambda=(\lambda_1,\lambda_2,\lambda_3)$.
The EHM generalizes the almost Mathieu operator by including a broader range of hopping terms. Crucially, it preserves an Aubry–André duality similar to that of the AMO \cite{Fan2018, JKS2005}, More importantly, the EHM possesses a stable self-dual region in its parameter space in which the spectrum is purely singular continuous—a feature absent in the AMO\cite{AJM2017,HanJitomirskaya2017}.

This raises the central question addressed in this work: Can one construct a unitary analogue of the extended Harper's model, that exhibits a stable parameter region with purely singular continuous spectrum, thus generalizing the UAMO in the same way the EHM generalizes the AMO? We will answer this question in the present research.

\subsection{Model and result}
The unitary almost Mathieu operator (UAMO) is represented as a generalized extended CMV matrix (GECMV)—a bi-infinite unitary matrix whose entries allow the usual CMV coefficients to carry complex phases. Let us explain this in details. 

Let $\mathcal{H}$ be the Hilbert space $\ell^{2}(\mathbb{Z})\otimes\mathbb{C}^{2}$ and $\{\delta_n^s\}$ be a standard basis as following
	$$\delta_{n}^{s}:=\delta_{n}\otimes e_{s},s\in\{+,-\},$$
	where $\{\delta_{n}\}$ is the standard basis of $\ell^{2}(\mathbb{Z})$ and $e_{+}=[1,0]^{\top},e_{-}=[0,1]^{\top}$ is the basis of $\mathbb{C}^{2}$. For an element $\psi\in\mathcal{H}$, denote its coordinate of $\delta_{n}^{s}$ by $\psi_{n}^{s}$ and let $\psi_{n}:=[\psi_{n}^{+},\psi_{n}^{-}]^{\top}$.
	We are concerned with a quantum walk $W:\ell^{2}(\mathbb{Z})\otimes\mathbb{C}^{2}\to \ell^{2}(\mathbb{Z})\otimes\mathbb{C}^{2}$. It has the form $W=S_{\lambda}\cdot Q$, where $S_{\lambda}$ is the conditional shift operator defined by
\begin{equation}\label{eq.shiftOp}S_{\lambda}\delta_{n}^{\pm}=\lambda\delta_{n\pm1}^{\pm}\pm\lambda'\delta_{n}^{\mp},\, \lambda'=\sqrt{1-\lambda^{2}},\end{equation}
	and the {\it coin} operator is given by
	\begin{equation}\label{eq.coinOp}[Q\psi]_{n}=Q_{n}\psi_{n}, \, Q_{n}:=
	\begin{pmatrix}
		q_{n}^{11} & q_{n}^{12} \\
		q_{n}^{21} & q_{n}^{22}
	\end{pmatrix}\in\mathbb{U}(2).\end{equation}

  The matrix representation of $W_\lambda$ under the basis $\{\delta_n^\pm\}$ up to a reordering is a generalized extended CMV matrix (GECMV)\footnote{A GECMV differs from an extended CMV matrix (ECMV) from that the associated coefficients $\rho_n$'s can be complex. While in an ECMV, $\rho_n=\sqrt{1-|\alpha_n|^2}$ is real for all $n.$} $\mathcal{E}$ that can be built as follows.
Let $\Theta_n=\begin{pmatrix}\overline{\alpha_n}&\rho_n\\\overline{\rho_n}&-\alpha_n\end{pmatrix}, n\in \bbZ$  acting on $\ell^{2}(\{n,n+1\})$ be the building block and $\{\alpha_n\}_{n\in\bbZ}\subset \mathbb{D}$ are the  {\it Verblunsky coefficients} and $\{\rho_{n}\}_{n\in\bbZ}\subset \bbC$  satisfy the condition $|\alpha_n|^2+|\rho_n|^2=1$. Denote $\mathcal{L}=\oplus_{n\in\bbZ}\Theta_{2n}$ and $\mathcal{M}=\oplus_{n\in\bbZ}\Theta_{2n+1}$, then $\mathcal{E}=\mathcal{L}\mathcal{M}$. In the case that $\rho_n=\sqrt{1-|\alpha_n|^2}$, $\mathcal{E}$ is called by Simon \cite{OPUC2} an {\it extended CMV matrix}.  The matrix expression of $\mathcal{E}$ takes the form
\begin{equation} \label{eq:gecmv}
	\mathcal E 	= 
        \begin{pmatrix}
		\ddots & \ddots & \ddots & \ddots &&&&  \\
		& \overline{\alpha_0\rho_{-1}} & \boxed{-\overline{\alpha_0}\alpha_{-1}} & \overline{\alpha_1}\rho_0 & \rho_1\rho_0 &&&  \\
		& \overline{\rho_0\rho_{-1}} & -\overline{\rho_0}\alpha_{-1} & {-\overline{\alpha_1}\alpha_0} & -\rho_1 \alpha_0 &&&  \\
		&&  & \overline{\alpha_2\rho_1} & -\overline{\alpha_2}\alpha_1 & \overline{\alpha_3} \rho_2 & \rho_3\rho_2 & \\
		&& & \overline{\rho_2\rho_1} & -\overline{\rho_2}\alpha_1 & -\overline{\alpha_3}\alpha_2 & -\rho_3\alpha_2 &    \\
		&& && \ddots & \ddots & \ddots & \ddots &
	\end{pmatrix},
\end{equation}
where we boxed the $(0,0)$ matrix element of $\mathcal{E}$.

The GECMV $\mathcal{E}$ and the quantum walk operator $W_\lambda$ are related by
identifying $\ell^2(\bbZ)\otimes\bbC^2$ with $\ell^2(\bbZ)$ by ordering the basis $\{\delta_n^s:n\in\bbZ,s\in\{+,-\}\}$ as \begin{equation}
\dots,\delta_{-1}^-,\delta_{0}^+,\delta_{0}^-,\delta_{1}^+,\delta_{1}^-,\delta_{2}^+,\dots,
\end{equation}
and by setting
\begin{equation}\label{GECMV}
	Q_{n}:=\begin{pmatrix}
		\overline{\rho_{2n-1}}&-\alpha_{2n-1}\\
		\overline{\alpha_{2n-1}}&\rho_{2n-1}
	\end{pmatrix}, \qquad (\alpha_{2n},\rho_{2n}):=(\lambda',\lambda).
\end{equation}

Such model was considered by \cite{FOZ17,CFO23} with the coin specified by \[Q_n=\begin{pmatrix}
    \lambda_2\cos(2\pi(\theta+n\Phi))+i\lambda_2'&-\lambda_2\sin(2\pi(\theta+n\Phi))\\\lambda_2\sin(2\pi(\theta+n\Phi))&\lambda_{2}\cos(2\pi(\theta+n\Phi))-i\lambda_2'
\end{pmatrix}\]
and was dubbed ``unitary almost Mathieu operator" (UAMO) by \cite{FOZ17}.

	Let $\theta\in\T$ and $\Phi$ be irrational.  We say $\Phi\in DC$ or satisfies the {\it Diophantine condition} if 
    there exists $\kappa>0,\tau=\tau(\Phi)>1$ such that for all nonzero $n\in\bbZ$
    \begin{equation}\label{eq.dio}\inf_{j\in\bbZ}\Vert j-n\Phi\Vert_{\bbR/\bbZ}\geq \frac{\kappa}{|n|^\tau}.
    \end{equation}
We say $\theta$ is irrational with respect to $\Phi$ if for all $j\in\mathbb{Z}$, $2\theta+j\Phi\notin\mathbb{Z}$.
    
    Consider a quasi-periodic example of quantum walk $W_{\lambda_1,\lambda_2,\theta,\Phi}=S_{\lambda_{1}}Q_{\lambda_{2}},\lambda_{1},\lambda_{2}\in(0,1)$ with the coin operator specified as the following 	\[q_{n}^{11}=f(\theta+n\Phi),q_{n}^{12}=-g(\theta+n\Phi),\]
	where
	\begin{equation}\label{eq.FuncF}f(x)=\frac{1}{k}(t\lambda_{2}e^{2\pi ix}+t\lambda_{2}e^{-2\pi ix}+(t^{2}-1)\lambda_{2}')=\frac{1}{k}(2t\lambda_{2}\cos{2\pi x}+(t^{2}-1)\lambda_{2}'),\end{equation}
    \begin{equation}\label{eq.FuncG}g(x)=\frac{1}{k}(-t^{2}\lambda_{2}e^{2\pi ix}+\lambda_{2}e^{-2\pi ix}+2t\lambda_{2}')\end{equation}
	 are functions defined on $\mathbb{T}$ and $|t|<1$, $k=-1-t^{2}$, $\lambda_{2}^{2}+\lambda_{2}'^{2}=1$.  
    
    The main purpose of this work is to present spectral properties of this model that resembles more the extended Harper's model considered by \cite{AJM2017} than the AMO. To be more specific, let us introduce our main results.
    
\begin{theorem}\label{thm.main}
Consider the quantum walk $W_{\lambda_1,\lambda_2,\theta,\Phi}$ associated with the conditional shift $S_{\lambda_1}$ and dynamically defined coins \eqref{eq.coinOp} with sampling functions \eqref{eq.FuncF} and \eqref{eq.FuncG}. Assume $t\neq 0$. Then the following hold:
\begin{enumerate}
\item If $\lambda_{1},\lambda_{2}\in\!\left[\frac{|1-t^{2}|}{1+t^{2}},\,1\right)$, then $W_{\lambda_1,\lambda_2,\theta,\Phi}$ has purely singular continuous spectrum for every irrational $\Phi$ and all but countably many $\theta$.
\item If $\lambda_{1}=\lambda_{2}$, then $W_{\lambda_1,\lambda_2,\theta,\Phi}$ has purely singular continuous spectrum for every irrational $\Phi$ and all but countably many $\theta$.
\end{enumerate}

If either $\lambda_{1}$ or $\lambda_{2}$ lies in $\left(0,\frac{|1-t^{2}|}{1+t^{2}}\right)$ and $\Phi\in DC$, then:
\begin{enumerate}[label=(\alph*)]
\item \label{ItemLocalization}If $\lambda_{1}<\lambda_{2}$, the quantum walk $W_{\lambda_1,\lambda_2,\theta,\Phi}$ exhibits Anderson localization and has pure point spectrum for almost every $\theta$.
\item \label{ItemAc} If $\lambda_{1}>\lambda_{2}$, the operator $W_{\lambda_1,\lambda_2,\theta,\Phi}$ has purely absolutely continuous spectrum for all $\theta$.
\end{enumerate}
\end{theorem}

The theorem gives the phase diagram demonstrated in FIGURE 1 below.

\begin{remark}
    The primary novelty of this work is the identification of the parameter region $[\frac{\left|1-t^{2}\right|}{1+t^{2}},1) \times [\frac{\left|1-t^{2}\right|}{1+t^{2}},1)$ (for $t\neq 0$) where the spectrum is purely singular continuous. This constitutes a stable critical parameters region, absent in the previously studied UAMO.
\end{remark}

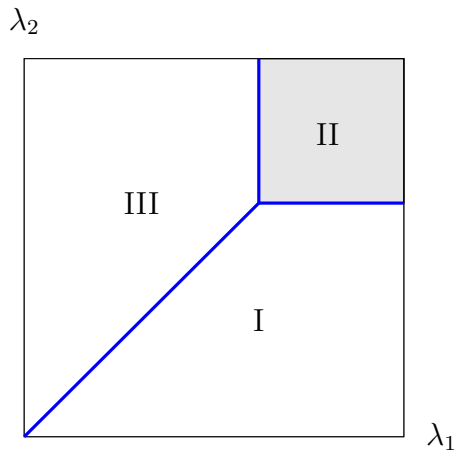
\begin{figure}\label{pict}
\centering
\begin{tikzpicture}[scale=5]

\draw (0,0) rectangle (1,1);

\def\s{0.618}

\draw[fill=white] (0,0) rectangle (1,1);
\draw[fill=gray!20] (\s,\s) rectangle (1,1);

\draw (\s,\s) rectangle (1,1);

\draw[blue, very thick] (\s,\s) -- (1,\s); 
\draw[blue, very thick] (\s,\s) -- (\s,1); 

\draw[blue, very thick] (0,0) -- (\s,\s);

\node at ({\s},{\s/2}) {\Rmnum{1}}; 
\node at ({\s/2},{\s}) {\Rmnum{3}}; 

\node at (0.8,0.8) {\Rmnum{2}};
\node at (1.1,0) {$\lambda_1$};
\node at (0,1.1) {$\lambda_2$};
\label{Fig1}
\end{tikzpicture}

\caption{In area \Rmnum{1}, the spectral type is purely absolutely continuous. In area \Rmnum{2}  and on the blue lines, the spectral type is purely singular continuous. In area \Rmnum{3}, the spectral type is purely point.}
\end{figure}

The case $t=0$ is of special interest. Moreover, the Lyapunov exponent can be computed (c. f. Corollary \ref{cor.simpleLE}) as \[\max\left\{\log\left(\frac{\lambda_2}{\lambda_1}\frac{\lambda_1'}{\lambda_2'}\right),0\right\}.\]
The quantity $\lambda=\frac{\lambda_2}{\lambda_1}\frac{\lambda_1'}{\lambda_2'}$ plays the role of the coupling constant as in the potential $2\lambda\cos 2\pi\theta$ of the AMO. 
\begin{theorem}\label{main2}
For $t=0$, that is, when $f(x)=\lambda_2'$ and $g(x)=-\lambda_{2} e^{-2\pi i x}$, the following statements hold:
\begin{enumerate}
\item If $0<\lambda_{1}=\lambda_{2}<1$, then $W_{\lambda_1,\lambda_2,\theta,\Phi}$ has purely singular continuous spectrum for every irrational $\Phi$ and all but countably many $\theta$.
\item \label{ItemLocalization2} If $0<\lambda_{1}<\lambda_{2}<1$, then $W_{\lambda_1,\lambda_2,\theta,\Phi}$ exhibits Anderson localization and has pure point spectrum for all $\Phi\in DC$ and for almost every $\theta$.
\item \label{ItemAc2} If $1>\lambda_{1}>\lambda_{2}>0$, then $W_{\lambda_1,\lambda_2,\theta,\Phi}$ has purely absolutely continuous spectrum for all $\Phi\in DC$ and all $\theta$.
\end{enumerate}
\end{theorem}

\begin{remark}
 In the special case 
$t=0$, the model 
$W_{\lambda_1,\lambda_2,\theta,\Phi}$	
  displays spectral transitions analogous to those of the UAMO, yet it is structurally distinct. Notably, the resulting operator takes the simpler form of an extended CMV matrix rather than a generalized extended CMV matrix.
\end{remark}

We recall the well-known Rakhmanov lemma concerning the absence of an absolutely continuous component in the spectral measure for CMV matrices. It states that if  
$\limsup_{n\to\infty} |\alpha_n| = 1,$  
then the associated CMV matrix has no absolutely continuous spectrum. This result can be found in Simon's book [see Theorems 4.3.4 and 10.9.7 in~\cite{OPUC2}].
The analogous version for Jacobi matrices, proved by \cite{Dombrowski1978}, states that for a singular $(\liminf_{n\to\infty} |a_n| = 0)$ Jacobi operator of the form  
\[
(Hu)_n = a_n u_{n+1} + \overline{a_{n-1}} u_{n-1} + v_n u_n,
\]   
implies the absence of absolutely continuous spectrum.
Our Theorem~\ref{thm.main} remains of interest in this context. In contrast to the typical setting where the conditional shift operator $S_{\lambda}$ is identified with the Laplacian term (which often introduces zeros in the hopping terms), our result does not rely on such vanishing coefficients. This distinction may be attributed to the additional freedom provided by the internal degree of freedom at each site.

\section{Preliminaries}
\subsection{Quasiperiodic Cocycle}
Given an irrational number $\Phi$ and a continuous map $M:\mathbb{T}\to M(2,\mathbb{C})$, we define the quasi-periodic \emph{cocycle}
\[
(\Phi,M):\mathbb{T}\times\mathbb{C}^{2}\to\mathbb{T}\times\mathbb{C}^{2},\qquad
(x,v)\mapsto (x+\Phi,\, M(x)v).
\]
The $n$-step iterates of this cocycle are given by $
(\Phi,M)^{n}=(n\Phi,\, M^{n}),$
where $M^0=\mathrm{id}$, and
\[
M^{n}(x):=
\begin{cases}
    \displaystyle\prod_{j=n-1}^{0} M(x+j\Phi), & n\ge 1,\\
    \displaystyle\prod_{j=n}^{-1} M^{-1}(x+j\Phi), & n\le -1.
\end{cases}
\]
The \emph{Lyapunov exponent} of the cocycle $(\Phi,M)$ is defined by
\[
L(\Phi,M)=\lim_{n\to\infty}\frac{1}{n}\int_{\mathbb{T}}
\ln \|M^{n}(x)\|\,dx.
\]
This limit is known to exist by Kingman's Subadditive Ergodic Theorem.
Let $C_\delta^\omega(\T,\mathbb{M}(2,\bbC))$ be the set of cocycle maps that have analytic extensions to $\{x+iy:x\in\T, |y|<\delta\}$ for some $\delta>0$ and $C^\omega(\T,\mathbb{M}(2,\bbC))=\cup_{\delta>0}C^\omega_\delta(\T,\mathbb{M}(2,\bbC))$. We call $(\Phi,M)$ an analytic cocycle if $\Phi\in\bbR\setminus\bbQ$ and $M\in C^{\omega}(\T,\mathbb{M}(2,\bbC))$.
\begin{theorem}[\cite{JM12}, \cite{BourgainJito}]\label{lem6.2}
Let $\Phi\in\bbR\setminus\bbQ$ be irrational. Then the Lyapunov exponent
\[
L(\Phi,M): C^{\omega}(\mathbb{T},\mathbb{M}(2,\mathbb{C}))\longrightarrow \mathbb{R}\cup\{-\infty\}
\]
is continuous in $M\in C^\omega(\T,\mathbb{M}(2,\bbC))$.
\end{theorem}

If $M\in C^0(\mathbb{T}, \mathrm{SL}(2,\mathbb{C}))$, we say that the cocycle $(\Phi,M)$ is \emph{uniformly hyperbolic} if there exist constants $c_{1},c_{2}>0$ and an continuous invariant splitting 
$\mathbb{C}^{2}=E^{s}(x)\oplus E^{u}(x), $
such that:
\begin{enumerate}
    \item $M(x)E^{*}(x)=E^{*}(x+\Phi)$ for $*=s,u$;
    \item For all $n\in\mathbb{Z}_{+}$, $x\in\mathbb{T}$, $\xi_{u}\in E^{u}(x+n\Phi)$, and $\xi_{s}\in E^{s}(x)$,
    \[
    \|M^{n}(x)\xi_{s}\|\le c_{1}e^{-nc_{2}}\|\xi_{s}\|,
    \qquad 
    \|M^{-n}(x)\xi_{u}\|\le c_{1}e^{-nc_{2}}\|\xi_{u}\|.
    \]
\end{enumerate}
We denote by $\mathcal{UH}$ the set of all uniformly hyperbolic cocycles.

\subsection{Avila's global theory}
Let $M:\mathbb{T}\to \mathbb{M}(2,\mathbb{C})$ be analytic with a holomorphic extension to the strip
$\mathbb{T}_{\delta}=\{x+i\epsilon : x\in\mathbb{T},\, |\epsilon|<\delta\}.$
For $|\epsilon|<\delta$, we consider the complexified cocycle $M(x+i\epsilon):\mathbb{T}\to \mathbb{M}(2,\mathbb{C})$ and define its Lyapunov exponent by
$
L(\Phi,M,\epsilon):=L(\Phi,M(\,\cdot+i\epsilon\,)).$
Since $M$ is holomorphic on $\mathbb{T}_{\delta}$, the convexity of the function $\epsilon\mapsto L(\Phi,M,\epsilon)$ for $|\epsilon|<\delta$ follows from the Hadamard three-lines theorem together with the subharmonicity of Lyapunov exponents. Consequently, $L(\Phi,M,\eta)$ admits right derivatives for all $|\eta|<\delta$, and we define
\[
\omega(\Phi,M,\eta)
:=\lim_{\epsilon\to0^{+}}
\frac{1}{2\pi\epsilon}\big(L(\Phi,M,\eta+\epsilon)-L(\Phi,M,\eta)\big)
\]
if the limit exists.
The quantity $\omega(\Phi,M,\eta)$ is called the \emph{acceleration} of the Lyapunov exponent at~$\eta$.

\begin{theorem}[\cite{Avila2015Acta}]\label{lem5.1}
Let $\Phi$ be irrational and let $M:\mathbb{T}\to M(2,\mathbb{C})$ be analytic with a holomorphic extension to the strip $\mathbb{T}_{\delta}=\{x+i\epsilon : x\in\mathbb{T},\, |\epsilon|<\delta\}$. Then 
$\omega(\Phi,M,\eta)\in \tfrac{1}{2}\mathbb{Z}$ for all $|\eta|<\delta.$
Moreover, if $M:\mathbb{T}\to\mathrm{SL}(2,\mathbb{C})$, then
$\omega(\Phi,M,\eta)\in \mathbb{Z}$ for all $|\eta|<\delta.$
\end{theorem}

Lyapunov exponents of complexified (in terms of phase) cocycles also characterize uniform hyperbolicity:
\begin{theorem}[\cite{Avila2015Acta}]\label{lem.UHGlobal}
Let $M:\mathbb{T}\to\mathrm{SL}(2,\mathbb{C})$ be analytic.  
The cocycle $(\Phi,M)$ is uniformly hyperbolic if and only if  
$L(\Phi,M)>0$ and 
$L(\Phi,M,\epsilon)$ is affine in a neighborhood of $\epsilon=0.$
\end{theorem}

In \cite{Avila2015Acta}, Avila  classified the spectral parameters in the spectrum of a one dimensional Schr\"odinger operator or equivalently, Schr\"odinger cocycles that are not uniformly hyperbolic into three regimes. This can be generalized to general $SL(2,\bbR)$ and $SU(1,1)$ valued quasiperiodic cocycles. A cocycle $(\Phi,M)$ that is not uniformly hyperbolic is said to be 
\begin{enumerate}
    \item \emph{supercritical} if $L(\Phi,M)>0$;
    \item \emph{subcritical} if $L(\Phi,M)=0$ and there exists $\epsilon>0$ such that $\omega(\Phi,M,\eta)=0$ for all $|\eta|<\epsilon$;
    \item \emph{critical} otherwise.
\end{enumerate}

\subsection{Szeg\H{o} Cocycle}
To study the spectral properties of the generalized extended CMV matrix $\mathcal{E}$ in \eqref{eq:gecmv}, one naturally considers the generalized eigenvalue equation 
\begin{equation*}
    \mathcal{E}\psi=z\psi, \qquad z\in\bbC.
\end{equation*}
Solutions to this equation satisfy the following recurrence:
\begin{equation}\label{eq.solIter1}
	\begin{pmatrix}\psi_{n+1}^{+}\\\psi_{n}^{-}\end{pmatrix}=A_{n,z}\begin{pmatrix}\psi_{n}^{+}\\\psi_{n-1}^{-}\end{pmatrix}, \quad n \in \bbZ,
\end{equation}
where the matrix $A_{n,z}$ is given by
\begin{equation}\label{eq:GECMV_transmat}
	A_{n,z}=\frac{1}{\rho_{2n}\rho_{2n-1}}\begin{pmatrix}
		z^{-1}+\alpha_{2n}\overline{\alpha}_{2n-1}+\alpha_{2n-1}\overline{\alpha_{2n-2}}+\alpha_{2n}\overline{\alpha_{2n-2}}z& -\overline{\rho_{2n-2}}\alpha_{2n-1}-\overline{\rho_{2n-2}}\alpha_{2n}z\\
		-\rho_{2n}\overline{\alpha_{2n-1}}-\rho_{2n}\overline{\alpha_{2n-2}}z&\rho_{2n}\overline{\rho_{2n-2}}z
	\end{pmatrix},
\end{equation}
for $n \in \bbZ$ and $z \in \bbC \setminus \{0\}$. 

Consider the generalized eigenvalue equation of the transposed GECMV matrix:
\begin{equation}\label{eq.dualEigen}
\mathcal{E}^\top\psi=z\psi.
\end{equation}
Since $\mathcal{L}$ and $\mathcal{M}$ are unitary, we have
$$z^{-1}\psi=(\mathcal{E}^{\top})^{-1}\psi=(\mathcal{M}^{\top}\mathcal{L}^{\top})^{-1}\psi=\overline{\mathcal{L}}\overline{\mathcal{M}}\psi=\overline{\mathcal{E}}\psi.$$
It thus follows from \eqref{eq.solIter1} that solutions $\psi\in \ell^{2}(\bbZ)$ to \eqref{eq.dualEigen} satisfy the following recurrence relation:
\begin{equation}\label{eq.transposeIter}
\begin{pmatrix}\psi_{n+1}^{+}\\\psi_{n}^{-}\end{pmatrix}=\overline{A_{n,\overline{z^{-1}}}}\begin{pmatrix}\psi_{n}^{+}\\\psi_{n-1}^{-}\end{pmatrix}, \quad n \in \bbZ.
\end{equation}
\medskip

For a general GECMV, the $\rho_n$'s in the building block may not be positive real. This causes additional issue in the analysis since the transfer matrix may not be in $\mathrm{SU}(1,1)$. This can be resolved by conjugating the GECMV to its associated ECMV via \cite[Theorem 2.1]{CFLOZ24}.
\begin{lemma}[\cite{CFLOZ24}]\label{lem.gaugeTransform}
   For any sequence of Verblunsky coefficients $\{\alpha_n\}$, there exists a family of generalized extended CMV matrices (GECMV) such that any two distinct matrices in the family are unitarily conjugate to each other by a gauge(diagonal) unitary transformation.  More precisely, for any two such matrices $\mathcal{E}(\alpha,\rho^{(1)})$ and $\mathcal{E}(\alpha,\rho^{(2)})$, one has
$\lvert \rho_n^{(1)} \rvert = \lvert \rho_n^{(2)} \rvert ,  \forall n\in\mathbb{Z}.$
In particular, within this unitary equivalence class of GECMV, there exists a canonical representative---the standard extended CMV matrix $\widetilde{\mathcal{E}}$---whose building blocks satisfy $\rho_n\in\mathbb{R}_{+}$ for all $n$.

\end{lemma}

Since after gauge transforming the $\rho_n$'s in \eqref{eq:GECMV_transmat} are positive, by \cite[Lemma 5.3]{CFLOZ24}, we have the following:
 \begin{equation}\label{eq.ConjugatedTransferMatrix}
    A_{n,z}=R_{2n}^{-1}JS^{+}_{n,z}JR_{2n-2},
\end{equation}
where $S^{+}_{n,z}=S_{2n,z}S_{2n-1,z}$ is determined by the \emph{normalized Szeg\H{o} transfer matrices}
\begin{equation}\label{eq:szego_normalized}
    S_{n,z}=\frac{z^{-\frac{1}{2}}}{\left|\rho_{n}\right|}\begin{pmatrix}z&-\overline{\alpha_{n}}\\-\alpha_{n}z&1\end{pmatrix} \in \SU(1,1),
\end{equation}
and
\begin{equation}\label{eq.Rn}
    R_{n}=\begin{pmatrix}1&0\\-\overline{\alpha_{n}}& \left|\rho_{n}\right| \end{pmatrix},\qquad J=\begin{pmatrix}0&1\\1&0\end{pmatrix}.
\end{equation}
\medskip

\begin{remark}
    As a general rule, the matrix representation of a random quantum walk imposes the restrictions that either $\alpha_n=\lambda,\rho_n=\sqrt{1-\lambda^2}$ for all odd $n\in\bbZ$ or for  all even $n\in\bbZ.$
    In the current context, the model we are interested in satisfies $(\alpha_{2n},\rho_{2n})=(\lambda',\lambda)$ as in \eqref{GECMV}. This is very useful since the conjugacy \eqref{eq.ConjugatedTransferMatrix} will be constant. Therefore we will mainly focus on the study of the two-step combined Szeg\H{o} cocycle $(\Phi,S^+_z)$.
\end{remark}

\section{Aubry-Andr\'e duality}
This section is devoted to establishing Aubry-Andr\'e duality, which serves as the foundation of the entire paper.
For a quantum walk $W_{\lambda}=S_{\lambda}Q$, let $W_{\lambda}^{\top}=Q^{\top}S_{\lambda}^{\top}$ be its transpose. Let $\psi\in\mathcal{H}$, then the actions of $W_\lambda$ and $W_\lambda^\top$  on $\psi$ can be computed explicitly as follows:
\begin{lemma}\label{lem.eigenRelation}
    For each $n\in\bbZ$, we have
    \[[W_\lambda\psi]_{n}^{+}=\lambda(q_{n-1}^{11}\psi_{n-1}^{+}+q_{n-1}^{12}\psi_{n-1}^{-})-\lambda'(q_{n}^{21}\psi_{n}^{+}+q_{n}^{22}\psi_{n}^{-}),\]
\[[W_\lambda\psi]_{n}^{-}=\lambda(q_{n+1}^{21}\psi_{n+1}^{+}+q_{n+1}^{22}\psi_{n+1}^{-})+\lambda'(q_{n}^{11}\psi_{n}^{+}+q_{n}^{12}\psi_{n}^{-}),\]
\[[W_\lambda^{\top}\psi]_{n}^{+}=q_{n}^{11}(\lambda\psi_{n+1}^{+}+\lambda'\psi_{n}^{-})+q_{n}^{21}(-\lambda'\psi_{n}^{+}+\lambda\psi_{n-1}^{-}),\]
\[[W_\lambda^{\top}\psi]_{n}^{-}=q_{n}^{12}(\lambda\psi_{n+1}^{+}+\lambda'\psi_{n}^{-})+q_{n}^{22}(-\lambda'\psi_{n}^{+}+\lambda\psi_{n-1}^{-}).\]
\end{lemma}
\begin{proof}
Direct computations from the definitions of the shift $S_{\lambda}$ and the coin $Q$.
\end{proof}

For $W_{\lambda_{1},\lambda_{2},\Phi,\theta}=S_{\lambda_{1}}Q_{\lambda_2}$ with the coin $Q_{\lambda_2}$ defined as 
 \begin{equation}\label{eq.generalForm}\left\{\begin{aligned}
&q_{n}^{11}=\overline{\rho_{2n-1}}=f(\theta+n\Phi)\\
&q_{n}^{12}=-\alpha_{2n-1}=-g(\theta+n\Phi)
\end{aligned}\right.\end{equation}
where $f$ and $g$ are given by \eqref{eq.FuncF} and \eqref{eq.FuncG} respectively.

\begin{lemma}\label{3.1}
  Let $z\in\partial\mathbb{D}$, for any $\psi\in \mathcal{H}$ satisfying $W_{\lambda_{1},\lambda_{2},\Phi,\theta}\psi=z\psi$, we have $$W_{\lambda_{2},\lambda_{1},\Phi,\xi}^{\top}\varphi^{\xi}=z\varphi^{\xi}$$
  for a.e. $\xi\in\mathbb{T}$,
  where $\varphi^{\xi}$ is defined as
  $$\begin{pmatrix}
  \varphi_{n}^{\xi,+} \\
  \varphi_{n}^{\xi,-}
  \end{pmatrix}=e^{2\pi in\theta}
  \begin{pmatrix}
  a & b\\
  b & -a
  \end{pmatrix}
  \begin{pmatrix}
  \check{\psi}^{+}(\xi+n\Phi) \\
  \check{\psi}^{-}(\xi+n\Phi)
  \end{pmatrix},$$
 $a=\frac{1}{\sqrt{1+t^{2}}},b=\frac{t}{\sqrt{1+t^{2}}}$ and $\check{.}$ means the inverse Fourier transform.
  \end{lemma}
 We call $W_{\lambda_{2},\lambda_{1},\Phi,\xi}^{\top}$ the {\it Aubry dual} of $W_{\lambda_{1},\lambda_{2},\Phi,\theta}$.
\begin{proof}
Denote the Fourier coefficients of $f,g$ by $f_{m},g_{m}$. 
For any $\psi\in \mathcal{H}$ satisfying $W_{\lambda_{1},\lambda_{2},\Phi,\theta}\psi=z\psi$, we have
  $$z\psi_{n}^{+}=\lambda_{1}(q_{n-1}^{11}\psi_{n-1}^{+}+q_{n-1}^{12}\psi_{n-1}^{-})-\lambda_{1}'(q_{n}^{21}\psi_{n}^{+}+q_{n}^{22}\psi_{n}^{-}),$$
  $$z\psi_{n}^{-}=\lambda_{1}(q_{n+1}^{21}\psi_{n+1}^{+}+q_{n+1}^{22}\psi_{n+1}^{-})+\lambda_{1}'(q_{n}^{11}\psi_{n}^{+}+q_{n}^{12}\psi_{n}^{-})$$
  by Lemma \ref{lem.eigenRelation}.
  Taking the inverse Fourier transformation of both sides, we have
  \begin{align*}
    z\check{\psi}^{+}(x)&=\sum_{n\in\mathbb{Z}}{e^{2\pi i(n+1)x}\lambda_{1}(f(\theta+n\Phi)\psi_{n}^{+}-g(\theta+n\Phi)\psi_{n}^{-})}
    \\&-\sum_{n\in\mathbb{Z}}{e^{2\pi inx}\lambda_{1}'(\overline{g}(\theta+n\Phi)\psi_{n}^{+}+\overline{f}(\theta+n\Phi)\psi_{n}^{-})}
    \\&=\sum_{n\in\mathbb{Z}}{e^{2\pi i(n+1)x}\lambda_{1}\sum_{m=-1}^{1}({f_{m}e^{2\pi im(\theta+n\Phi)}\psi_{n}^{+}-g_{m}e^{2\pi im(\theta+n\Phi)}\psi_{n}^{-}})}
    \\&-\sum_{n\in\mathbb{Z}}{e^{2\pi inx}\lambda_{1}'\sum_{m=-1}^{1}({\overline{g_{-m}}e^{2\pi im(\theta+n\Phi)}\psi_{n}^{+}+\overline{f_{-m}}e^{2\pi im(\theta+n\Phi)}\psi_{n}^{-}})}
    \\&=\sum_{m=-1}^{1}{\lambda_{1} e^{2\pi i(x+m\theta)}(f_{m}\check{\psi}^{+}(x+m\Phi)-g_{m}\check{\psi}^{-}(x+m\Phi))}
    \\&-\sum_{m=-1}^{1}{\lambda_{1}' e^{2\pi im\theta}(\overline{g_{-m}}\check{\psi}^{+}(x+m\Phi)+\overline{f_{-m}}\check{\psi}^{-}(x+m\Phi))}
    \\&=\sum_{m=-1}^{1}{(\lambda_{1} e^{2\pi i(x+m\theta)}f_{m}-\lambda_{1}'e^{2\pi im\theta}\overline{g_{-m}})\check{\psi}^{+}(x+m\Phi)}
    \\&-\sum_{m=-1}^{1}{(\lambda_{1} e^{2\pi i(x+m\theta)}g_{m}+\lambda_{1}'e^{2\pi im\theta}\overline{f_{-m}})\check{\psi}^{-}(x+m\Phi)}
    \\&=(\frac{t}{k}\lambda_{1}\lambda_{2}e^{2\pi i(x-\theta)}+\frac{t^{2}}{k}\lambda_{1}'\lambda_{2}e^{-2\pi i\theta})\check{\psi}^{+}(x-\Phi)+(\frac{t^{2}-1}{k}\lambda_{1}\lambda_{2}'e^{2\pi ix}-\frac{2t}{k}\lambda_{1}'\lambda_{2}')\check{\psi}^{+}(x)
    \\&+(\frac{t}{k}\lambda_{1}\lambda_{2}e^{2\pi i(x+\theta)}-\frac{1}{k}\lambda_{1}'\lambda_{2}e^{2\pi i\theta})\check{\psi}^{+}(x+\Phi)
    \\&-(\frac{1}{k}\lambda_{1}\lambda_{2}e^{2\pi i(x-\theta)}+\frac{t}{k}\lambda_{1}'\lambda_{2}e^{-2\pi i\theta})\check{\psi}^{-}(x-\Phi)-(\frac{2t}{k}\lambda_{1}\lambda_{2}'e^{2\pi ix}+\frac{t^{2}-1}{k}\lambda_{1}'\lambda_{2}')\check{\psi}^{-}(x)
    \\&+(\frac{t^{2}}{k}\lambda_{1}\lambda_{2}e^{2\pi i(x+\theta)}-\frac{t}{k}\lambda_{1}'\lambda_{2}e^{2\pi i\theta})\check{\psi}^{-}(x+\Phi)
  \end{align*}
  and
  \begin{align*}
    z\check{\psi}^{-}(x)&=\sum_{m=-1}^{1}{(\lambda_{1} e^{2\pi i(-x+m\theta)}\overline{g_{-m}}+\lambda_{1}'e^{2\pi im\theta}f_{m})\check{\psi}^{+}(x+m\Phi)}
    \\&+\sum_{m=-1}^{1}{(\lambda_{1} e^{2\pi i(-x+m\theta)}\overline{f_{-m}}-\lambda_{1}'e^{2\pi im\theta}g_{m})\check{\psi}^{-}(x+m\Phi)}
    \\&=(-\frac{t^{2}}{k}\lambda_{2}\lambda_{1}e^{2\pi i(-x-\theta)}+\frac{t}{k}\lambda_{1}'\lambda_{2}e^{-2\pi i\theta})\check{\psi}^{+}(x-\Phi)+(\frac{2t}{k}\lambda_{1}\lambda_{2}'e^{-2\pi ix}+\frac{t^{2}-1}{k}\lambda_{1}'\lambda_{2}')\check{\psi}^{+}(x)
    \\&+(\frac{1}{k}\lambda_{1}\lambda_{2}e^{2\pi i(-x+\theta)}+\frac{t}{k}\lambda_{1}'\lambda_{2}e^{2\pi i\theta})\check{\psi}^{+}(x+\Phi)
    \\&+(\frac{t}{k}\lambda_{1}\lambda_{2}e^{2\pi i(-x-\theta)}-\frac{1}{k}\lambda_{1}'\lambda_{2}e^{-2\pi i\theta})\check{\psi}^{-}(x-\Phi)+(\frac{t^{2}-1}{k}\lambda_{1}\lambda_{2}'e^{-2\pi ix}-\frac{2t}{k}\lambda_{1}'\lambda_{2}')\check{\psi}^{-}(x)
    \\&+(\frac{t}{k}\lambda_{1}\lambda_{2}e^{2\pi i(-x+\theta)}+\frac{t^{2}}{k}\lambda_{1}'\lambda_{2}e^{2\pi i\theta})\check{\psi}^{-}(x+\Phi).
  \end{align*}
 
  By the definition of $\varphi^{\xi}$, by multiplying both sides by $z$, we obtain
  $$z\begin{pmatrix}
  \varphi_{n}^{\xi,+} \\
  \varphi_{n}^{\xi,-}
  \end{pmatrix}=e^{2\pi in\theta}\frac{1}{\sqrt{1+t^2}}
  \begin{pmatrix}
  1 & t\\
  t & -1
  \end{pmatrix}
  \begin{pmatrix}
  z\check{\psi}^{+}(\xi+n\Phi) \\
  z\check{\psi}^{-}(\xi+n\Phi)
  \end{pmatrix}.$$
  On the other hand, as shown in Lemma \ref{lem.eigenRelation}, we have
  \begin{align*}
[W_{\lambda_{2},\lambda_{1},\Phi,\xi}^{\top}\varphi^{\xi}]_{n}^{+}&=q_{n}^{11}(\lambda_{2}\varphi_{n+1}^{\xi,+}+\lambda_{2}'\varphi_{n}^{\xi,-})+q_{n}^{21}(-\lambda_{2}'\varphi_{n}^{\xi,+}+\lambda_{2}\varphi_{n-1}^{\xi,-})
    \\&=e^{2\pi in\theta}[\frac{1}{k}(t\lambda_{1}e^{2\pi i(\xi+n\Phi)}+t\lambda_{1}e^{-2\pi i(\xi+n\Phi)}+(t^{2}-1)\lambda_{1}')(\lambda_{2}e^{2\pi i\theta}(a\check{\psi}^{+}(\xi+(n+1)\Phi)
    \\&+b\check{\psi}^{-}(\xi+(n+1)\Phi))+\lambda_{2}'(b\check{\psi}^{+}(\xi+n\Phi)-a\check{\psi}^{-}(\xi+n\Phi)))
    \\&+\frac{1}{k}(-t^{2}\lambda_{1}e^{-2\pi i(\xi+n\Phi)}+\lambda_{1}e^{2\pi i(\xi+n\Phi)}+2t\lambda_{1}')(-\lambda_{2}'(a\check{\psi}^{+}(\xi+n\Phi)+b\check{\psi}^{-}(\xi+n\Phi))
    \\&+\lambda_{2}e^{-2\pi i\theta}(b\check{\psi}^{+}(\xi+(n-1)\Phi)-a\check{\psi}^{-}(\xi+(n-1)\Phi)))]
  \end{align*}
  and
  \begin{align*}
    [W_{\lambda_{2},\lambda_{1},\Phi,\xi}^{\top}\varphi^{\xi}]_{n}^{-}&=q_{n}^{12}(\lambda_{2}\varphi_{n+1}^{\xi,+}+\lambda_{2}'\varphi_{n}^{\xi,-})+q_{n}^{22}(-\lambda_{2}'\varphi_{n}^{\xi,+}+\lambda_{2}\varphi_{n-1}^{\xi,-})
     \\&=e^{2\pi in\theta}[\frac{1}{k}(t^{2}\lambda_{1}e^{2\pi i(\xi+n\Phi)}-\lambda_{1}e^{-2\pi i(\xi+n\Phi)}-2t\lambda_{1}')(\lambda_{2}e^{2\pi i\theta}(a\check{\psi}^{+}(\xi+(n+1)\Phi)
    \\&+b\check{\psi}^{-}(\xi+(n+1)\Phi))+\lambda_{2}'(b\check{\psi}^{+}(\xi+n\Phi)-a\check{\psi}^{-}(\xi+n\Phi)))
    \\&+\frac{1}{k}(t\lambda_{1}e^{2\pi i(\xi+n\Phi)}+t\lambda_{1}e^{-2\pi i(\xi+n\Phi)}+(t^{2}-1)\lambda_{1}')(-\lambda_{2}'(a\check{\psi}^{+}(\xi+n\Phi)+b\check{\psi}^{-}(\xi+n\Phi))
    \\&+\lambda_{2}e^{-2\pi i\theta}(b\check{\psi}^{+}(\xi+(n-1)\Phi)-a\check{\psi}^{-}(\xi+(n-1)\Phi)))].
  \end{align*}
  For any fixed $n$, we now compare the coefficients of $\check{\psi}^{\pm}(\xi+n\Phi+m\Phi)(m=-1,0,1)$ in the expressions of $z\varphi_{n}^{\xi,\pm}$ and $[W_{\lambda_{2},\lambda_{1},\Phi,\xi}^{\top}\varphi^{\xi}]_{n}^{\pm}$ above. By direct calculation, we note that these coefficients are equal, hence $W_{\lambda_{2},\lambda_{1},\Phi,\xi}^{\top}\varphi^{\xi}=z\varphi^{\xi}$. 
\end{proof}

\section{Lyapunov exponents and parameter regions of the cocycle}
In this section, we study the dynamics of the generalized eigenvalue equation $W_{\lambda_1,\lambda_2,\Phi,\theta} \psi=z\psi$. We will briefly describe the iterative relation $\psi$ obeys following the lines of \cite{CFLOZ24}. Such relation can be conjugated to a two-step combined parametrized Szeg\H{o} cocycles which are frequently used in the study of CMV matrices.  We will calculate the Lyapunov exponents of the Szeg\H{o} cocycles in different parameter regions. With the explicit expressions of the Lyapunov exponents, we are able divide the parameter region for $(\lambda_1,\lambda_2)$ into different types according to Avila's global theory.

\subsection{Szeg\H{o} cocycles maps} Let $\alpha_{2n-1}=g(\theta+n\Phi),\rho_{2n-1}=\overline{f(\theta+n\Phi)}=f(\theta+n\Phi)$ and $\alpha_{2n}=\lambda_1',\rho_{2n}=\lambda_1$ for $n\in\bbZ$. Let $W_{\lambda_1,\lambda_2,\Phi,\theta}$ be the associated quantum walk operator. The transfer matrix $A_{n,z}$ in \eqref{eq:GECMV_transmat} can be realized as $A_{n,z}=A_z(\theta+n\Phi)$ where 
\begin{equation}\label{eq.cocycleMap}
A_{z}(x)=
\begin{pmatrix}
 A_{z}(x)_{11}  & A_{z}(x)_{12}\\[2mm]
 A_{z}(x)_{21}  & A_{z}(x)_{22}
\end{pmatrix},
\end{equation}
where the entries of $A_{z}(x)$ take the following explicit form:
\begin{align*}
A_{z}(x)_{11}
&= \frac{1}{f(x)}
\Bigl[
\lambda_{1}^{-1}z^{-1}
+\frac{1}{k}\lambda_{1}'\lambda_{1}^{-1}\bigl((1-t^{2})\lambda_{2}(e^{2\pi ix}+e^{-2\pi ix})+4t\lambda_{2}'\bigr)
+z\,\lambda_{1}'^{2}\lambda_{1}^{-1}
\Bigr], \\[1mm]
A_{z}(x)_{12}
&= \frac{1}{f(x)}
\Bigl[
-\tfrac{1}{k}\bigl(-t^{2}\lambda_{2}e^{2\pi ix}+\lambda_{2}e^{-2\pi ix}+2t\lambda_{2}'\bigr)
-\lambda_{1}' z
\Bigr], \\[1mm]
A_{z}(x)_{21}
&= \frac{1}{f(x)}
\Bigl[
-\tfrac{1}{k}\bigl(-t^{2}\lambda_{2}e^{-2\pi ix}+\lambda_{2}e^{2\pi ix}+2t\lambda_{2}'\bigr)
-\lambda_{1}' z
\Bigr], \\[1mm]
A_{z}(x)_{22}
&= \frac{\lambda_{1} z}{f(x)}.
\end{align*}
Recall that $f(x)$ takes the form   
$f(x)=\frac{1}{k}\bigl(2t\lambda_{2}\cos(2\pi x)+(t^{2}-1)\lambda_{2}'\bigr).$
Consequently, by \eqref{eq.ConjugatedTransferMatrix} and \eqref{eq:szego_normalized},  the two-step combined Szeg\H{o} cocycle $(\Phi,S_{z}^{+}$, where
$S_{z}^{+}(x)=\frac{z^{-1}}{\lambda_{1}\left|f(x)\right|}D_{z}(x)$ and $D_z$ is given by the following 
$$D_{z}(x):=
\begin{pmatrix}
   z^{2}-\lambda_{1}'z\frac{1}{k}(t^{2}\lambda_{2}e^{2\pi ix}-\lambda_{2}e^{-2\pi ix}-2t\lambda_{2}') & \frac{1}{k}(t^{2}\lambda_{2}e^{-2\pi ix}-\lambda_{2}e^{2\pi ix}-2t\lambda_{2}')z-\lambda_{1}' \\
   -\lambda_{1}'z^{2}+\frac{1}{k}(t^{2}\lambda_{2}e^{2\pi ix}-\lambda_{2}e^{-2\pi ix}-2t\lambda_{2}')z & -\lambda_{1}'z\frac{1}{k}(t^{2}\lambda_{2}e^{-2\pi ix}-\lambda_{2}e^{2\pi ix}-2t\lambda_{2}')+1
\end{pmatrix}$$
In what follows, we compute the Lyapunov exponents of the  cocycle $(\Phi,S_{z}^{+})$.  
For clarity, we divide the analysis into two distinct regimes:
\begin{enumerate}
    \item \emph{Non-singular case:} $\lambda_{2}\in\bigl(0,\tfrac{|1-t^{2}|}{\,1+t^{2}\,}\bigr)$.  
    In this regime, $f(x)\neq 0$ for all $x\in\mathbb{T}$.
    \item \emph{Singular case:} $\lambda_{2}\in\bigl[\tfrac{|1-t^{2}|}{\,1+t^{2}\,},1\bigr)$.  
    Here, the function $f$ admits zeros.
\end{enumerate}

\subsection{Nonsingular case}
\label{sec5.1}
In the case $\lambda_{2}\in(0,\frac{|1-t^{2}|}{1+t^{2}})$, 
 the Szeg\H{o} cocycle is  analytic with a holomorphic extension to a strip $\mathbb{T}_{\epsilon_{0}}=\{x+i\epsilon:x\in\mathbb{T},\left|\epsilon\right|<\epsilon_{0}\}$ for some $\epsilon_{0}>0$.  And for $\left|\epsilon\right|<\epsilon_{0}$, 
\begin{align*}
 L(\Phi,S_{z}^{+},\epsilon)=\lim_{n\to\infty}{\frac{1}{n}\int_{\mathbb{T}}{\ln{\left\|S_{z}^{+,n}(x+i\epsilon)\right\|}}dx}
 =L(\Phi,D_{z},\epsilon)-\ln{\lambda_{1}}-\int_{\mathbb{T}}{\ln{\left|f(x+i\epsilon)\right|}dx},
\end{align*}
where the last equality follows from the ergodic theorem applied to $\ln{\left|f(x+i\epsilon)\right|}\in L^{1}(\mathbb{T})$. 

We first calculate the integral $\int_{\mathbb{T}}{\ln{\left|f(x+i\epsilon)\right|}dx}$.
\begin{lemma}\label{lem5.3.1}
We have
    $$\int_{\mathbb{T}}{\ln{\left|f(x+i\epsilon)\right|}dx}=
     \ln{\frac{1}{2}(\left|t^{2}-1\right|\lambda_{2}'+\sqrt{(t^{2}+1)^{2}\lambda_{2}'^{2}-4t^{2}})}-\ln{(-k)}, \quad \left|\epsilon\right|<\epsilon_{0}$$
where $\epsilon_{0}=\frac{1}{2\pi}\ln{\frac{\left|t^{2}-1\right|\lambda_{2}'+\sqrt{(t^{2}+1)^{2}\lambda_{2}'^{2}-4t^{2}}}{2t\lambda_{2}}}\in(0,\infty]$ for $\lambda_{2}\in(0,\frac{|1-t^{2}|}{1+t^{2}})$.
\end{lemma}
\begin{proof}
Recall that
\[
f(x+i\epsilon)
    = \omega^{-1}\frac{1}{k}\bigl(t\lambda_{2}e^{-2\pi \epsilon}\omega^{2}
        + t\lambda_{2}e^{2\pi\epsilon}
        + (t^{2}-1)\lambda_{2}'\omega\bigr)
    := z^{-1} f_{1}(\omega),
\]
where $\omega=e^{2\pi i x}$.  Consequently,
$\int_{\mathbb{T}} \ln |f(x+i\epsilon)|\, dx
    = \int_{\mathbb{T}} \ln |f_{1}(\omega)|\, dx .$
The function $f_{1}(\omega)$ is holomorphic on $\mathbb{C}$ and has two zeros, denoted $\omega_{1}$ and $\omega_{2}$, given  by
\[
\omega_{1,2}
    = \frac{(1-t^{2})\lambda_{2}' \pm \sqrt{(1-t^{2})^{2}\lambda_{2}'^{2}
        -4t^{2}\lambda_{2}^{2}}}{2t\lambda_{2}}
        \, e^{2\pi\epsilon}.
\]
Hence the integral can be computed via Jensen's formula.  
Depending on the value of $\epsilon$, a different number of zeros lie in the open unit disk $\mathbb{D}$. 

Now if  $|\epsilon|\le \epsilon_{0}$, exactly one of $\omega_{1},\omega_{2}$ lies in $\mathbb{D}$.  
Thus Jensen's formula gives
\[
\int_{\mathbb{T}} \ln |f_{1}(\omega)|\, dx
    = \ln |f_{1}(0)| - \ln |\omega_{1,2}|
    = \ln\!\left[\tfrac{1}{2}\bigl(|t^{2}-1|\lambda_{2}'
        + \sqrt{(t^{2}+1)^{2}\lambda_{2}'^{2}-4t^{2}}\bigr)\right]
      - \ln(-k).
\]

\end{proof}

We then compute the Lyapunov exponent $L(\Phi,D_{z},\epsilon)$ when $\epsilon\to\pm\infty$.
\begin{lemma}\label{lem5.5.1}
    For $t\neq0$ and any $z\in\partial\mathbb{D}$, when $\left|\epsilon\right|$ is large enough, we have
    $$L(\Phi,D_{z},\epsilon)=
    \begin{cases}
 \ln{(\left|t\right|\lambda_{1}\lambda_{2})}-\ln{(-k)}+2\pi\left|\epsilon\right|, \quad \lambda_{1}\in[\frac{|1-t^{2}|}{1+t^{2}},1),\\
 \ln{[\frac{\lambda_{2}}{2}(\left|1-t^{2}\right|\lambda_{1}'+\sqrt{(1+t^{2})^{2}\lambda_{1}'^{2}-4t^{2}})]}-\ln{(-k)}+2\pi\left|\epsilon\right|, \quad \lambda_{1}\in(0,\frac{|1-t^{2}|}{1+t^{2}}).
\end{cases}$$
\end{lemma}
\begin{proof}
    We first compute
    \begin{align*}
        \lim_{\epsilon\to+\infty}{e^{-2\pi\epsilon}D_{z}(x+i\epsilon)}&=
    \begin{pmatrix}
   \lambda_{1}'z\frac{1}{k}\lambda_{2}e^{-2\pi ix} & \frac{1}{k}t^{2}\lambda_{2}e^{-2\pi ix}z\\
   -\frac{1}{k}\lambda_{2}e^{-2\pi ix}z & -\lambda_{1}'z\frac{1}{k}t^{2}\lambda_{2}e^{-2\pi ix}
\end{pmatrix}\\
&=\frac{1}{k}ze^{-2\pi ix}
\begin{pmatrix}
     \lambda_{1}'\lambda_{2} & t^{2}\lambda_{2}\\
   -\lambda_{2} & -\lambda_{1}'t^{2}\lambda_{2}
\end{pmatrix}:=\frac{1}{k}ze^{-2\pi ix}A.
    \end{align*}
And the constant matrix $A$ has eigenvalues
$$e_{1},e_{2}=
 \begin{cases}
  \frac{\lambda_{2}}{2}((1-t^{2})\lambda_{1}'\pm\sqrt{4t^{2}-(1+t^{2})^{2}\lambda_{1}'^{2}}i), \lambda_{1}\in[\frac{|1-t^{2}|}{1+t^{2}},1),\\
 \frac{\lambda_{2}}{2}((1-t^{2})\lambda_{1}'\pm\sqrt{(1+t^{2})^{2}\lambda_{1}'^{2}-4t^{2}}), \lambda_{1}\in(0,\frac{|1-t^{2}|}{1+t^{2}}).
\end{cases}$$
Then by Theorem \ref{lem6.2}, 
\begin{align*}
    \lim_{\epsilon\to+\infty}{L(\Phi,e^{-2\pi\epsilon}D_{z}(x+i\epsilon))}&=L(\Phi,A)-\ln{(-k)}=\ln{\max\{\left|e_{1}\right|},\left|e_{2}\right|\}-\ln{(-k)}\\
    &=
    \begin{cases}
 \ln{(\left|t\right|\lambda_{1}\lambda_{2})}-\ln{(-k)}, \lambda_{1}\in[\frac{|1-t^{2}|}{1+t^{2}},1),\\
 \ln{[\frac{\lambda_{2}}{2}(\left|1-t^{2}\right|\lambda_{1}'+\sqrt{(1+t^{2})^{2}\lambda_{1}'^{2}-4t^{2}})]}-\ln{(-k)}, \lambda_{1}\in(0,\frac{|1-t^{2}|}{1+t^{2}}).
\end{cases}
\end{align*} 
By Theorem \ref{lem5.1}, for $\epsilon>0$ large enough, 
 $$L(\Phi,D_{z},\epsilon)=
     \begin{cases}
 \ln{(\left|t\right|\lambda_{1}\lambda_{2})}-\ln{(-k)}+2\pi\epsilon,\lambda_{1}\in[\frac{|1-t^{2}|}{1+t^{2}},1);\\
 \ln{[\frac{\lambda_{2}}{2}(\left|1-t^{2}\right|\lambda_{1}'+\sqrt{(1+t^{2})^{2}\lambda_{1}'^{2}-4t^{2}})]}-\ln{(-k)}+2\pi\epsilon,\lambda_{1}\in(0,\frac{|1-t^{2}|}{1+t^{2}}).
\end{cases}$$
The case $\epsilon<0$ can be dealt with by considering  
$ \lim_{\epsilon\to-\infty}{e^{2\pi\epsilon}D_{z}(x+i\epsilon)}$ in the same way.
\end{proof}

Now we are ready to give the Lyapunov exponent of the cocycle $(\Phi,S_{z}^{+})$.
\begin{theorem}\label{thm6.4}
    When $\lambda_{2}\in(0,\frac{|1-t^{2}|}{1+t^{2}})$, for $z\in\Sigma$, we have
$$ L(\Phi,S_{z}^{+})=\max\{F(\lambda_{1},\lambda_{2}),0\},$$
    where
    $$F(\lambda_{1},\lambda_{2})=\ln{\frac{\lambda_{2}(\left|1-t^{2}\right|\lambda_{1}'+\sqrt{(1+t^{2})^{2}\lambda_{1}'^{2}-4t^{2}})}{\lambda_{1}(\left|1-t^{2}\right|\lambda_{2}'+\sqrt{(1+t^{2})^{2}\lambda_{2}'^{2}-4t^{2}})}}.$$
    In particular, the Szeg\H{o} cocycle $(\Phi,S_{z}^{+}(x))$ is
\begin{enumerate}
    \item supercritical, if $\lambda_{1}<\lambda_{2}$;
    \item critical, if $\lambda_{1}=\lambda_{2}$;
    \item  subcritical, if $\lambda_{1}>\lambda_{2}$.
\end{enumerate}
\end{theorem}
\begin{proof}
    By Lemma  \ref{lem5.3.1}, when $\left|\epsilon\right|<\epsilon_{0}$, $\int_{\mathbb{T}}{\ln{\left|f(x+i\epsilon)\right|}dx}$ is a constant, hence $\omega(\Phi,D_{z},\epsilon)=\omega(\Phi,S_{z}^{+},\epsilon)\in\mathbb{Z}$ since $S_{z}^{+}(x)\in\mathbb{SU}(1,1)$. By the convexity of $L(\Phi,D_{z},\epsilon)$, $\omega(\Phi,D_{z},\epsilon)$ can only be $0$ or $\pm1$ by Lemma  \ref{lem5.5.1}.

  When \(\lambda_{1}\in(0,\frac{|1-t^{2}|}{1+t^{2}})\) and \(\epsilon\to \infty\), \[L(\Phi,D_{z},\epsilon)=\ln{\frac{\lambda_{2}}{2}(\left|1-t^{2}\right|\lambda_{1}'+\sqrt{(1+t^{2})^{2}\lambda_{1}'^{2}-4t^{2}})}-\ln{(-k)}+2\pi\left|\epsilon\right|.\] 
  For $\lambda_{1}\le\lambda_{2}$, \[\ln{\frac{\lambda_{2}}{2}(\left|1-t^{2}\right|\lambda_{1}'+\sqrt{(1+t^{2})^{2}\lambda_{1}'^{2}-4t^{2}})}-\ln{(-k)}\ge\ln{\frac{\lambda_{1}}{2}(\left|1-t^{2}\right|\lambda_{2}'+\sqrt{(1+t^{2})^{2}\lambda_{2}'^{2}-4t^{2}})}-\ln{(-k)},\] then
for $z\in\Sigma$, the cocycle $S_{z}^{+}(x)$ is not uniformly hyperbolic, so by Theorem \ref{lem.UHGlobal}, we have \[L(\Phi,D_{z},0)=\ln{\frac{\lambda_{2}}{2}(\left|1-t^{2}\right|\lambda_{1}'+\sqrt{(1+t^{2})^{2}\lambda_{1}'^{2}-4t^{2}})}-\ln{(-k)},\] hence $ L(\Phi,S_{z}^{+})=F(\lambda_{1},\lambda_{2})\ge0$ and $\omega(\Phi,D_{z},0)=\omega(\Phi,S_{z}^{+},0)=1$. Therefore, when $\lambda_{1}<\lambda_{2}$, $L(\Phi,S_{z}^{+})>0$ and then the Szeg\H{o} cocycle $(\Phi,S_{z}^{+}(x))$ is supercritical; when $\lambda_{1}=\lambda_{2}$, $L(\Phi,S_{z}^{+})=0$ and then $(\Phi,S_{z}^{+}(x))$ is critical.

For  $\lambda_{1}>\lambda_{2}$, \[\ln{\frac{\lambda_{2}}{2}(\left|1-t^{2}\right|\lambda_{1}'+\sqrt{(1+t^{2})^{2}\lambda_{1}'^{2}-4t^{2}})}-\ln{(-k)}<\ln{\frac{\lambda_{1}}{2}(\left|1-t^{2}\right|\lambda_{2}'+\sqrt{(1+t^{2})^{2}\lambda_{2}'^{2}-4t^{2}})}-\ln{(-k)},\] then by the non-negativity of $L(\Phi,S_{z}^{+})$,
for $z\in\Sigma$, the cocycle $S_{z}^{+}(x)$ is not uniformly hyperbolic, so by Theorem \ref{lem.UHGlobal}, we have 
\[L(\Phi,D_{z},0)=\ln{\frac{\lambda_{1}}{2}(\left|1-t^{2}\right|\lambda_{2}'+\sqrt{(1+t^{2})^{2}\lambda_{2}'^{2}-4t^{2}})}-\ln{(-k)},\] hence $ L(\Phi,S_{z}^{+})=0$ and $\omega(\Phi,D_{z},0)=\omega(\Phi,S_{z}^{+},0)=0$. Therefore $(\Phi,S_{z}^{+}(x))$ is subcritical.

When $t\neq0$ and $\lambda_{1}\in[\frac{|1-t^{2}|}{1+t^{2}},1)$(hence $\lambda_{1}>\lambda_{2}$) and $\epsilon\to \infty$, $L(\Phi,D_{z},\epsilon)=\ln{\left|t\right|\lambda_{1}\lambda_{2}}-\ln{(-k)}+2\pi\left|\epsilon\right|$. Notice that 
\begin{align*}
    \ln{\left|t\right|\lambda_{1}\lambda_{2}}-\ln{(-k)}&\le\ln{\frac{\lambda_{2}}{2}(\left|1-t^{2}\right|\lambda_{1}'+\sqrt{(1+t^{2})^{2}\lambda_{1}'^{2}-4t^{2}})}-\ln{(-k)}\\
    &<\ln{\frac{\lambda_{1}}{2}(\left|1-t^{2}\right|\lambda_{2}'+\sqrt{(1+t^{2})^{2}\lambda_{2}'^{2}-4t^{2}})}-\ln{(-k)}.
\end{align*}
Then by the non-negativity of $L(\Phi,S_{z}^{+})$, for $z\in\Sigma$, the cocycle $S_{z}^{+}(x)$ is not uniformly hyperbolic, so by Theorem \ref{lem.UHGlobal}, we have $L(\Phi,D_{z},0)=\ln{\frac{\lambda_{1}}{2}(\left|1-t^{2}\right|\lambda_{2}'+\sqrt{(1+t^{2})^{2}\lambda_{2}'^{2}-4t^{2}})}-\ln{(-k)}$, hence $ L(\Phi,S_{z}^{+})=0$ and $\omega(\Phi,D_{z},0)=\omega(\Phi,S_{z}^{+},0)=0$. Therefore $(\Phi,S_{z}^{+}(x))$ is subcritical. 

In conclusion, $L(\Phi,S_{z}^{+})=\max\{F(\lambda_{1},\lambda_{2}),0\},$ when $z\in\Sigma$.
\end{proof}

\begin{lemma}\label{cor.simpleLE}
For $t=0,z\in\Sigma$, we have \[L(\Phi,S_z^+)=\max\{\log\left(\frac{\lambda_1'\lambda_2}{\lambda_1\lambda_2'}\right),0\}.\]
 In particular, the Szeg\H{o} cocycle $(\Phi,S_{z}^{+})$ is
\begin{enumerate}
    \item supercritical, if $\lambda_{1}<\lambda_{2}$;
    \item critical, if $\lambda_{1}=\lambda_{2}$;
    \item  subcritical, if $\lambda_{1}>\lambda_{2}$.
\end{enumerate}
\end{lemma}

\begin{proof}
In this case, we observe that
$\int_{\mathbb{T}} \ln\!\left| f(x+i\epsilon) \right| \, dx 
    = \ln \lambda_{2}'$ for all  $\epsilon\in\mathbb{R}$ and $\lambda_{2}\in(0,1),$
since $f(x+i\epsilon)\equiv -\lambda_{2}'$.  
Moreover, in complete analogy with Lemma~\ref{lem5.5.1}, when $\lvert \epsilon\rvert$ is sufficiently large we obtain
$L(\Phi,D_{z},\epsilon)
    = \ln(\lambda_{1}'\lambda_{2}) + 2\pi \lvert \epsilon\rvert.$
The desired conclusion follows immediately.

\end{proof}

\subsection{Singular case}

Recall that if $t \neq 0$ and $\lambda_{2} \in \bigl[ \frac{|1-t^{2}|}{1+t^{2}}, 1 \bigr)$, the function  
$f(x)$  
has zeros on $\mathbb{T}$, making the Szegő cocycle singular and non‑holomorphic.  
Thus Avila's global theory cannot be applied directly.
To bypass this, we introduce a parameter $r$ near $t^{2}$ (with $t$ fixed) and use Theorem \ref{lem6.2} to approximate the singular cocycle by non‑singular ones,  
taking the limit $r \to t^{2}$.
When $\lambda_{2} \neq \frac{|1-t^{2}|}{1+t^{2}}$, define perturbed sampling functions  
\[
f_{r}(x)=\frac{1}{k}\Bigl( \frac{r}{t}\lambda_{2}e^{2\pi i x}
      + t\lambda_{2}e^{-2\pi i x}+(r-1)\lambda_{2}'\Bigr),
\qquad
g_{r}(x)=\frac{1}{k}\Bigl( -r\lambda_{2}e^{2\pi i x}
      + \lambda_{2}e^{-2\pi i x}+\bigl(t+\tfrac{r}{t}\bigr)\lambda_{2}'\Bigr),
\]
where $r$ is real, $|r-t^{2}|$ is small, and  
$k=-\sqrt{1+r^{2}+r^{2}/t^{2}+t^{2}}<0$.  
Then $|f_{r}(x)|^{2}+|g_{r}(x)|^{2}\equiv 1$.  
Denote by $\Sigma_{r}$ the spectrum of the GECMV matrix with $f_{r},g_{r}$.

For $r \neq t^{2}$, $f_{r}(x)$ vanishes on $\mathbb{T}$ exactly when  
$\frac{r}{t}\lambda_{2}+t\lambda_{2}= \pm (r-1)\lambda_{2}'$.  
Since $t^{2}+1\neq 0$, at most two such $r$ exist. Hence for some $\delta>0$ and all  
$r \in (t^{2}-\delta,t^{2})\cup(t^{2},t^{2}+\delta)$, we have $f_{r}(x)\neq 0$ on $\mathbb{T}$.
For such $r$, the Szegő cocycle is  
\[
S_{z,r}^{+}(x)=\frac{z^{-1}}{\lambda_{1}|f_{r}(x)|}\,D_{z,r}(x),
\]
where $$D_{z,r}:=
\begin{pmatrix}
   z^{2}-\lambda_{1}'z\frac{1}{k}(r\lambda_{2}e^{2\pi ix}-\lambda_{2}e^{-2\pi ix}-(t+\frac{r}{t})\lambda_{2}') & \frac{1}{k}(r\lambda_{2}e^{-2\pi ix}-\lambda_{2}e^{2\pi ix}-(t+\frac{r}{t})\lambda_{2}')z-\lambda_{1}' \\
   -\lambda_{1}'z^{2}+\frac{1}{k}(r\lambda_{2}e^{2\pi ix}-\lambda_{2}e^{-2\pi ix}-(t+\frac{r}{t})\lambda_{2}')z & -\lambda_{1}'z\frac{1}{k}(r\lambda_{2}e^{-2\pi ix}-\lambda_{2}e^{2\pi ix}-(t+\frac{r}{t})\lambda_{2}')+1
\end{pmatrix}$$ is entire in $x$. To extend $S_{z,r}^{+}$ analytically to a strip  
$\mathbb{T}_{\epsilon_{0}}=\{x+i\epsilon: x\in\mathbb{T}, |\epsilon|<\epsilon_{0}\}$, we only need  to make holomorphic extension of  $|f_{r}(x)|$. Since $f_{r}(x)\neq0$ on  
$\mathbb{T}$, $F_{r}(w)=f_{r}(w)\overline{f_{r}(\overline{w})}$ is non‑vanishing in a strip $\mathbb{T}_{\epsilon_{0}}$ and  
admits an analytic square root. Choosing the branch that is positive on $\mathbb{T}$,
we set  
\begin{equation}\label{extension}
|f_{r}|_{\text{ext}}(x+i\epsilon)=\sqrt{f_{r}(x+i\epsilon)\overline{f_{r}(x-i\epsilon)}}.
\end{equation}
Then $|f_{r}|_{\text{ext}}(x)=|f_{r}(x)|$, and $|f_{r}|_{\text{ext}}$ is analytic on  
$\mathbb{T}_{\epsilon_{0}}$. Consequently, $S_{z,r}^{+}$ extends holomorphically to  
$\mathbb{T}_{\epsilon_{0}}$ as  
\[
S_{z,r}^{+}(x+i\epsilon)=\frac{1}{\lambda_{1}z\,|f_{r}|_{\text{ext}}(x+i\epsilon)}
      D_{z,r}(x+i\epsilon)\in\mathrm{SL}(2,\mathbb{C}).
\]

In order to compute the Lyapunov exponents, we need the following result.

\begin{lemma}\label{lem:integral}
There exists $\delta>0$ such that for all $r\in(t^{2}-\delta,t^{2}+\delta)$, with
$
\epsilon_{0}:=\frac{1}{2\pi}\ln\frac{\sqrt{|r|}}{|t|},$
the following holds.
\begin{enumerate}
    \item If $\epsilon_{0}<0$, then
\begin{equation}\label{iep'}
\int_{\mathbb{T}}\ln |f_{r}|_{\mathrm{ext}}(x+i\epsilon)\,dx
= \ln\bigl(|t|\lambda_{2}\bigr)-\ln(-k),
\qquad |\epsilon|\le -\epsilon_{0}.
\end{equation}
\item If $\epsilon_{0}\ge 0$, then
\begin{equation}\label{iep''}
\int_{\mathbb{T}}\ln |f_{r}|_{\mathrm{ext}}(x+i\epsilon)\,dx
= \ln\!\bigl(\bigl|\tfrac{r}{t}\bigr|\lambda_{2}\bigr)-\ln(-k),
\qquad |\epsilon|\le \epsilon_{0}.
\end{equation}
\end{enumerate}
In particular, $\int_{\mathbb{T}}\ln |f_{r}|_{\mathrm{ext}}(x+i\epsilon)\,dx$ is continuous with respect to $r$ at $r=t^{2}$. 
\end{lemma}

\begin{proof}
Set $\omega=e^{2\pi i x}$. Then
$f_{r}(x+i\epsilon)=\omega^{-1}f_{1,r}(\omega),$
where
\[
f_{1,r}(\omega)
=\frac{1}{k}
\Bigl(\tfrac{r}{t}\lambda_{2}e^{-2\pi\epsilon}\omega^{2}
+(r-1)\lambda_{2}'\omega
+t\lambda_{2}e^{2\pi\epsilon}\Bigr).
\]
Hence
\[
\int_{\mathbb{T}}\ln|f_{r}(x+i\epsilon)|\,dx
=\int_{\mathbb{T}}\ln|f_{1,r}(\omega)|\,dx.
\]

For $r$ sufficiently close to $t^{2}$, the discriminant
\(
(1-r)^{2}\lambda_{2}'^{2}-4r\lambda_{2}^{2}
\)
is negative, so the two zeros of $f_{1,r}$ are
\[
\omega_{\pm}
=\frac{(1-r)\lambda_{2}'\pm
i\sqrt{4r\lambda_{2}^{2}-(1-r)^{2}\lambda_{2}'^{2}}}
{2\frac{r}{t}\lambda_{2}}\;e^{2\pi\epsilon}.
\]
A direct computation shows that $|\omega_{\pm}|<1$ if and only if $\epsilon<\epsilon_{0}$.

\smallskip
\noindent
\emph{If $\epsilon<\epsilon_{0}$}, both zeros lie in $\mathbb{D}$ and Jensen's formula yields
\[
\int_{\mathbb{T}}\ln|f_{r}(x+i\epsilon)|\,dx
=\ln|f_{1,r}(0)|-\ln|\omega_{+}|-\ln|\omega_{-}|
=\ln\!\bigl(\bigl|\tfrac{r}{t}\bigr|\lambda_{2}\bigr)-\ln(-k)-2\pi\epsilon .
\]

\smallskip
\noindent
\emph{If $\epsilon\ge\epsilon_{0}$}, no zeros of $f_r$ lies in $\mathbb{D}$, hence
\[
\int_{\mathbb{T}}\ln|f_{r}(x+i\epsilon)|\,dx
=\ln|f_{1,r}(0)|
=\ln\bigl(|t|\lambda_{2}\bigr)-\ln(-k)+2\pi\epsilon .
\]

Finally, by the definition of the analytic extension \(|f_{r}|_{\mathrm{ext}}\) (see \eqref{extension}), the linear terms $\pm 2\pi\epsilon$  cancel on the strips
\(|\epsilon|\le -\epsilon_{0}\) and \(|\epsilon|\le\epsilon_{0}\), respectively, yielding
\eqref{iep'} and \eqref{iep''}.
\end{proof}

Secondly, similar to Lemma \ref{lem5.5.1}, we have
\begin{lemma}
There exists $\delta>0$ such that for all $r\in(t^{2}-\delta,t^{2}+\delta)$,
\[
L(\Phi,D_{z,r},\epsilon)=
\begin{cases}
\ln(\sqrt{|r|}\,\lambda_{1}\lambda_{2})-\ln(-k)+2\pi|\epsilon|, 
& \lambda_{1}\in\bigl(\frac{|1-t^{2}|}{1+t^{2}},\,1\bigr),\\[0.3em]
\ln\!\Bigl[\dfrac{\lambda_{2}}{2}\Bigl(|1-r|\lambda_{1}'
+\sqrt{(1+r)^{2}{\lambda_{1}'^{2}}-4r}\Bigr)\Bigr]-\ln(-k)+2\pi|\epsilon|,
& \lambda_{1}\in\bigl(0,\,\frac{|1-t^{2}|}{1+t^{2}}\bigr),
\end{cases}
\]
as $\epsilon\to\pm\infty$.
\end{lemma}

\begin{proof}
We start by computing the normalized limit of the complexified cocycle:
$\lim_{\epsilon\to+\infty} e^{-2\pi\epsilon}D_{z,r}(x+i\epsilon)
   = \frac{1}{k}ze^{-2\pi ix}A_{r},$
where
\[
A_{r}=
\begin{pmatrix}
\lambda_{1}'\lambda_{2} & r\lambda_{2}\\[0.2em]
-\lambda_{2} & -\lambda_{1}'r\lambda_{2}
\end{pmatrix}.
\]
A direct computation shows that the eigenvalues of $A_{r}$ are
\[
e_{1,2}=
\begin{cases}
\displaystyle
\frac{\lambda_{2}}{2}\!\left((1-r)\lambda_{1}'
\pm i\sqrt{\,4r-(1+r)^{2}\lambda_{1}'^{2}\,}\right),
& \lambda_{1}\in\left(\frac{|1-t^{2}|}{1+t^{2}},\,1\right),\\[0.8em]
\displaystyle
\frac{\lambda_{2}}{2}\!\left((1-r)\lambda_{1}'
\pm\sqrt{(1+r)^{2}\lambda_{1}'^{2}-4r}\right),
& \lambda_{1}\in\left(0,\,\frac{|1-t^{2}|}{1+t^{2}}\right).
\end{cases}
\]
Because $r$ is restricted to a small neighborhood of $t^{2}$, the signs of the discriminants remain fixed, so the above expressions are valid uniformly.

By Theorem~\ref{lem6.2} and  Theorem~\ref{lem5.1}, the desired result follows,
and the formula for $\epsilon<0$ follows analogously by considering 
$\lim_{\epsilon\to-\infty} e^{2\pi\epsilon}D_{z,r}(x+i\epsilon)$.
\end{proof}

By following the same lines as in the proof of Theorem \ref{thm6.4}, recalling that $\Sigma_r$ denotes the spectrum of the perturbed model with $t^2$ replaced by $r$ in a neighborhood of $t^2$, we have the following result:
\begin{lemma}\label{lem6.9}
    For $r\in(t^{2}-\delta,t^{2})\cup(t^{2},t^{2}+\delta)$, $\lambda_{1}\neq\frac{\left|1-t^{2}\right|}{1+t^{2}}$, $\lambda_{2}\in(\frac{|1-t^{2}|}{1+t^{2}},1)$ and $z\in\Sigma_{r}$, we have
$$ L(\Phi,S_{z,r}^{+})=\max\{G(r,\lambda_{1},\lambda_{2}),0\},$$
where
$$G(r,\lambda_{1},\lambda_{2})=L(\Phi,D_{z,r},0)-I_{r}(0)-\ln{\lambda_{1}}.$$
\end{lemma}
To conclude the Lyapunov exponent of $S_z^+=S_{z,t^{2}}^{+}$, we also need the continuity of spectrum with respect to $r$. The lemma below follows from Theorem 1.2 of \cite{BT2025}.
\begin{lemma}[\cite{BT2025}]\label{lem6.10}
    The spectrum $\sum_{r}$ of the quantum walk $W_{\lambda_{1},\lambda_{2},\Phi,\theta,r}$ is continuous with respect to $r$ near $t^{2}$ in  Hausdorff metric. 
\end{lemma}
We are ready to compute the Lyapunov exponent of the singular cocycle.
\begin{theorem}\label{thm6.11}
    For any $\lambda_{1}\in(0,1)$, $\lambda_{2}\in[\frac{|1-t^{2}|}{1+t^{2}},1)$ and $z\in\Sigma$, we have
    $$ L(\Phi,S_{z}^{+})=
    \begin{cases}
    0, \lambda_{1}\in[\frac{|1-t^{2}|}{1+t^{2}},1);\\
     \ln{\frac{1}{2\left|t\right|\lambda_{1}}(\left|1-t^{2}\right|\lambda_{1}'+\sqrt{(1+t^{2})^{2}\lambda_{1}'^{2}-4t^{2}})}, \lambda_{1}\in(0,\frac{|1-t^{2}|}{1+t^{2}}).
    \end{cases}$$
\end{theorem}
\begin{proof}
    If $\lambda_{1}\neq\frac{\left|1-t^{2}\right|}{1+t^{2}}$, $\lambda_{2}\in(\frac{|1-t^{2}|}{1+t^{2}},1)$ and $z\in\Sigma$. When $r\to t^{2}$, by Lemma \ref{lem6.10}, there exist $z_{r}\in\sum_{r}$ such that $z_{r}\to z$. Then  $L(\Phi,S_{z_{r},r}^{+})\to L(\Phi,S_{z}^{+})$ by Theorem \ref{lem6.2} applied to $L(\Phi,D_{z,r})$.  By Lemma \ref{lem6.9},
$ L(\Phi,S_{z_{r},r}^{+})=\max\{G(r,\lambda_{1},\lambda_{2}),0\}$. Taking the limit of $G(r,\lambda_{1},\lambda_{2})$ as $r\to t^{2}$, we get the conclusion.

    If $\lambda_{1}=\frac{\left|1-t^{2}\right|}{1+t^{2}}$ or $\lambda_{2}=\frac{|1-t^{2}|}{1+t^{2}}$, we can approximate the model with $\lambda_{1}\to\frac{\left|1-t^{2}\right|}{1+t^{2}}$ or $\lambda_{2}\to\frac{\left|1-t^{2}\right|}{1+t^{2}}$ and use the same continuity argument above.
\end{proof}

\section{Absolutely continuous spectrum in the subcritical region}

First we have the following general result: 
\begin{theorem}\label{thmac}
Let $\Phi\in DC$ and consider the extended CMV matrix $\tilde{\mathcal{E}}(\theta)$ with Verblunsky coefficients given by
\begin{equation}\label{eq.evenOddCoeff}
\alpha_{2n-1}=g_1(\theta+n\Phi),\qquad
\alpha_{2n}=g_2(\theta+n\Phi),
\end{equation}
where $g_1,g_2:\mathbb{T}\to\mathbb{D}$ are analytic sampling functions satisfying $\sup_{x\in\T}|g_i(x)|<1, i=1,2$.
If the Szeg\H{o} cocycle associated with $\tilde{\mathcal{E}}(\theta)$ is subcritical for every $z\in\Sigma$, then $\tilde{\mathcal{E}}(\theta)$ has purely absolutely continuous spectrum for all $\theta\in\mathbb{T}$.
\end{theorem}

\begin{proof}
The coefficients specified by \eqref{eq.evenOddCoeff} is technically not quasiperiodic but almost periodic. This seemingly difficulty can be easily resolved by combining two consecutive steps. Indeed, if we deonte $S(\alpha,z)$ as the cocycle map associted to the sampling function $\alpha:\T\to\bbD$, then instead of working with the cocycles $(\Phi,S(g_1,z)),(\Phi,S(g_2,z))$, we work directly with $(2\Phi, S(g_2,z)S(g_1,z))$. The new system is a quasiperiodic cocycle with possibly doubled frequency $2\Phi$, compare \cite{LDZ22JFA}.

The strategy of the proof originates in Avila's work \cite{Avila2008}.
When the Verblunsky coefficients are in the perturbative regime, namely when the sampling functions are sufficiently close to constants, the conclusion follows directly from \cite[Theorem~1.1]{LDZ22TAMS}.
In the nonperturbative regime (where the smallness condition is independent of the frequency $\Phi$), the argument is carried out in \cite[Appendix~B]{LDZ22TAMS}.

To treat the subcritical regime, one replaces \cite[Theorem~B.1]{LDZ22TAMS} by Avila's solution of the Almost Reducibility Conjecture (i.e. subcritical implies almost reducible)\cite{Avila23} .
Although in the present setting the extended CMV matrix is almost-periodic—since $\alpha_{2n-1}$ and $\alpha_{2n}$ are generated by two different analytic sampling functions—the argument carries over verbatim. The essential input is the refined growth control provided by almost reducibility, and no step of the proof relies on the coincidence of the sampling functions. For closely related arguments in the Schr\"odinger case, see, for instance, \cite[Theorem 4.1]{ZhouWang2023CMP}.
\end{proof}


\begin{proof}[Proof of Theorem \ref{thm.main}(b) and Theorem \ref{main2}(3)]
In our setting, $\alpha_{2n}=\lambda'$ is constant, while $\alpha_{2n-1}=g(\theta+n\Phi)$ remains quasiperiodic.  
The conclusion therefore follows directly from Theorem~\ref{thm6.4}, Lemma~\ref{cor.simpleLE}, and Theorem~\ref{thmac}.

\end{proof}

\section{Singular continuous spectrum}

\subsection{Singular continuous spectrum in regime  \Rmnum{2}}
In this section, we prove purely singular continuous spectrum in regime  \Rmnum{2}, first we state the absence of absolutely continuous spectrum: 

\begin{lemma}\label{coro.noAc2}
    For $\lambda_{2}\in[\frac{|1-t^{2}|}{1+t^{2}},1)$, $W_{\lambda_{1},\lambda_{2},\Phi,\theta}$ has no absolutely continuous spectrum for every irrational $\Phi$ and  every $\theta\in\T$.
\end{lemma}
\begin{proof}
    First we generalize the well-known Rakhmanov lemma\cite{OPUC2} to GECMV:
    \begin{lemma}[\cite{OPUC2}]\label{lemma4.1}
  If $\limsup_{n\to\infty}\left|\alpha_{n}\right|=1$, then the associated GECMV has no absolutely continuous spectrum.
\end{lemma}
\begin{proof}
    In the case that $\mathcal{E}$ is an extended CMV matrix, let $\mathcal{C}$ be the standard CMV matrix associated to $\{\alpha_n\}_{n\geq 0}$. If $\limsup_{n\to+\infty}|\alpha_n|=1$, then by \cite[Theorem 10.9.7]{OPUC2}, $\mathcal{C}$ has no absolutely continuous spectrum. It is a well known result that the essential spectrum of $\mathcal{C}$ coincides with that of $\mathcal{E}$ and in each connected component of $\partial\mathbb{D}\setminus \sigma(\mathcal{E})$ there is up to one discrete spectrum of $\mathcal{C}$. This implies that if $\mathcal{C}$ has no absolutely continuous spectrum, then $\mathcal{E}$ has no absolutely continuous spectrum.

    In the case that $\mathcal{E}$ is a GECMV, then by Lemma \ref{lem.gaugeTransform} there exists a diagonal unitary matrix $\mathcal{U}=\mathcal{U}(\theta)$ such that $\tilde{\mathcal{E}}=\mathcal{U}\mathcal{E}\mathcal{U}^{-1}$ is an extended CMV matrix with the same Verblunsky coefficients. Apply the arguments above we obtain the lemma.
\end{proof}

Recall that $f(x)=\frac{1}{k}(2t\lambda_{2}\cos{2\pi x}+(t^{2}-1)\lambda_{2}')$, thus
$f$ has zeros on $\mathbb{T}$ if and only if $\lambda_{2}\in[\frac{|1-t^{2}|}{1+t^{2}},1)$. Then the result follows from Lemma \ref{lemma4.1} and the minimality of the linear flow $x\mapsto x+\Phi$. 
\end{proof}

It remains to show the absence of point spectrum, then Theorem \ref{thm.main}(1) follows directly  from  Lemma \ref{coro.noAc2} and Lemma \ref{prop.noEigen2}. 
\begin{lemma}\label{prop.noEigen2}
     For $\lambda_{1}\in [\frac{|1-t^{2}|}{1+t^{2}},1)$, $W_{\lambda_{1},\lambda_{2},\Phi,\theta}$ has no eigenvalue for all irrational $\Phi$ and $\theta$ if $\theta$ is irrational with respect to $\Phi$.
\end{lemma}
\begin{proof}
The main strategy follows the lines in \cite{AJM2017} and \cite{Hanrui2017}. 
We first note if  $t\lambda_{2}e^{2\pi ix}+t\lambda_{2}e^{-2\pi ix}+(t^{2}-1)\lambda_{2}'\neq0$, then we have  
the following symmetry of cocycle: 
\begin{equation} \label{lem.antiSymmetry} 
    R^{-1}A_{z}^{-1}(x)R=A_{z}(-x),\text{ where } R=\begin{pmatrix}
        0 & 1\\
        -1 & 0
    \end{pmatrix}.
\end{equation}

Now if $\lambda_{1}\in [\frac{|1-t^{2}|}{1+t^{2}},1)$, assume that $W_{\lambda_{1},\lambda_{2},\theta,\Phi}$ has an eigenvalue $z\in\partial\mathbb{D}$ and an eigenfunction $\psi\in \mathcal{H}$. By Lemma  \ref{3.1}, $W_{\lambda_{1},\lambda_{2},\theta,\Phi}$ has Aubry dual $W_{\lambda_{2},\lambda_{1},\xi,\Phi}^{\top}$. So for a.e. $\xi\in\mathbb{T}$, $\varphi^{\xi}$ is an eigenfunction of $W_{\lambda_{2},\lambda_{1},\xi,\Phi}^{\top}$ associated to the  eigenvalue $z$, where 
  $$\begin{pmatrix}
  \varphi_{n}^{\xi,+} \\
  \varphi_{n}^{\xi,-}
  \end{pmatrix}=e^{2\pi in\theta}
  \begin{pmatrix}
    a & b\\
    b & -a
  \end{pmatrix}
  \begin{pmatrix}
  \check{\psi}^{+}(\xi+n\Phi) \\
  \check{\psi}^{-}(\xi+n\Phi)
  \end{pmatrix}:=e^{2\pi in\theta}
  \begin{pmatrix}
  \check{\phi}^{+}(\xi+n\Phi)\\
  \check{\phi}^{-}(\xi+n\Phi)
  \end{pmatrix}.$$
  Therefore, by (\ref{eq.transposeIter}) we have the iterative relation by the transfer matrix of $W^\top_{\lambda_2,\lambda_1,\Phi,\xi}$
  \begin{align}\label{eq4.1}
    \begin{pmatrix}
    e^{2\pi i\theta}\check{\phi}^{+}(x+\Phi)\\
    \check{\phi}^{-}(x)
    \end{pmatrix}=
    \overline{A_{\lambda_{2},\lambda_{1},z}(x)}
    \begin{pmatrix}
    \check{\phi}^{+}(x)\\
    e^{-2\pi i\theta}\check{\phi}^{-}(x-\Phi)
    \end{pmatrix}
  \end{align}
  for a.e. $x\in\mathbb{T}$.
  Replacing $x$ with $-x$, we have
  $$\begin{pmatrix}
  e^{2\pi i\theta}\check{\phi}^{+}(-x+\Phi)\\
  \check{\phi}^{-}(-x)
  \end{pmatrix}=
  \overline{A_{\lambda_{2},\lambda_{1},z}(-x)}
  \begin{pmatrix}
  \check{\phi}^{+}(-x)\\
  e^{-2\pi i\theta}\check{\phi}^{-}(-x-\Phi)
  \end{pmatrix}.$$
  By   \eqref{lem.antiSymmetry} we obtain
  $$\begin{pmatrix}
  e^{2\pi i\theta}\check{\phi}^{+}(-x+\Phi)\\
  \check{\phi}^{-}(-x)
  \end{pmatrix}=
  R^{-1}\overline{A_{\lambda_{2},\lambda_{1},z}(x)}^{-1}R
  \begin{pmatrix}
  \check{\phi}^{+}(-x)\\
  e^{-2\pi i\theta}\check{\phi}^{-}(-x-\Phi)
  \end{pmatrix},$$
that is
  \begin{align}\label{eq4.2}
    \overline{A_{\lambda_{2},\lambda_{1},z}(x)}
    \begin{pmatrix}
    \check{\phi}^{-}(-x)\\
    -e^{2\pi i\theta}\check{\phi}^{+}(-x+\Phi)
    \end{pmatrix}=
    \begin{pmatrix}
    e^{-2\pi i\theta}\check{\phi}^{-}(-x-\Phi)\\
    -\check{\phi}^{+}(-x)
    \end{pmatrix}.
  \end{align}
  By  (\ref{eq4.1}), (\ref{eq4.2}), we have
  $$\overline{A_{\lambda_{2},\lambda_{1},z}(x)}
  \begin{pmatrix}
  \check{\phi}^{+}(x) & \check{\phi}^{-}(-x)\\
  e^{-2\pi i\theta}\check{\phi}^{-}(x-\Phi) & -e^{2\pi i\theta}\check{\phi}^{+}(-x+\Phi)
  \end{pmatrix}=
  \begin{pmatrix}
  e^{2\pi i\theta}\check{\phi}^{+}(x+\Phi) & e^{-2\pi i\theta}\check{\phi}^{-}(-x-\Phi)\\
  \check{\phi}^{-}(x) & -\check{\phi}^{+}(-x)
  \end{pmatrix}.$$
  Let 
  $M(x)=\begin{pmatrix}
  \check{\phi}^{+}(x) & \check{\phi}^{-}(-x)\\
  e^{-2\pi i\theta}\check{\phi}^{-}(x-\Phi) & -e^{2\pi i\theta}\check{\phi}^{+}(-x+\Phi)
  \end{pmatrix}$,
  then
  \begin{align}\label{eq8.3}
      \overline{A_{\lambda_{2},\lambda_{1},z}(x)}M(x)=M(x+\Phi)
  \begin{pmatrix}
  e^{2\pi i\theta} & 0\\
  0 & e^{-2\pi i\theta}
  \end{pmatrix}
  \end{align}
  for a.e. $x\in\mathbb{T}$.
  By the same argument in the proof of \cite{Hanrui2017} and using $\theta$ is $\Phi$-irrational we know $\left \|  M(x)\right \|>0$ and $\det{M(x)}\equiv C>0$ for a.e. $x\in\T$. 

  More precisely, first, if there exists a subset of $\mathbb{T}$ of positive Lebesgue measure where $M(x)=0$, then by (\ref{eq8.3}) and the ergodicity, $M(x)=0$ for a.e. $x\in\mathbb{T}$, thus $\check{\phi}^{\pm}(x)=0$ for a.e. $x\in\mathbb{T}$, which is a contradiction with $\psi\neq0$.

  Second, if there exists a subset of $\mathbb{T}$ of positive Lebesgue measure where $\det{M(x)}=0$, then by $\det{M(x)}=\det{M(x+\Phi)}$ obtained from (\ref{eq8.3}) and the ergodicity we conclude $\det{M(x)}=0$ for a.e. $x\in\mathbb{T}$. Hence for a.e. $x\in\mathbb{T}$, there exists $l(x)\in\mathbb{C}$ s.t. 
  $$\begin{pmatrix}
  \check{\phi}^{+}(x) \\
  e^{-2\pi i\theta}\check{\phi}^{-}(x-\Phi) 
  \end{pmatrix}=l(x)
  \begin{pmatrix}
  \check{\phi}^{-}(-x)\\
  -e^{2\pi i\theta}\check{\phi}^{+}(-x+\Phi)
  \end{pmatrix}.$$
Therefore, by (\ref{eq8.3}) we have
  $$l(x)\begin{pmatrix}
  \check{\phi}^{+}(x+\Phi) \\
  e^{-2\pi i\theta}\check{\phi}^{-}(x) 
  \end{pmatrix}=l(x+\Phi)e^{2\pi i\theta}\overline{A_{\lambda_{2},\lambda_{1},z}(x)}
  \begin{pmatrix}
  \check{\phi}^{+}(x) \\
  e^{-2\pi i\theta}\check{\phi}^{-}(x-\Phi) 
  \end{pmatrix}=l(x+\Phi)e^{4\pi i\theta}
  \begin{pmatrix}
  \check{\phi}^{+}(x+\Phi) \\
  e^{-2\pi i\theta}\check{\phi}^{-}(x) 
  \end{pmatrix}.$$
  Since for a.e. $x\in\mathbb{T}$, $M(x)\neq0$, we have $l(x)=e^{4\pi i\theta}l(x+\Phi)$. So $\left|l(x)\right|$ is a constant by ergodicity, hence $l(x)$ is in $L^{1}(\mathbb{T})$ by measurability. Thus we can expand $l(x)$ as a Fourier series: $l(x)=\sum\hat{l}_{j}e^{2\pi ijx}$ where
  $$\hat{l}_{j}=\hat{l}_{j}e^{4\pi i\theta}e^{2\pi ij\Phi}$$
  for all $j\in\mathbb{Z}$. Since $\hat{l}_{j}$ is not always zero, there exists $j$ such that $1=e^{2\pi i(2\theta+j\Phi)}$ and then $2\theta+j\Phi\in\mathbb{Z}$. However, we assumed $\theta$ is irrational with respect to $\Phi$, contradiction!
  
  Moreover, since entries $\check{\phi}^{\pm}$ of $M$ are all in $L^{2}(\mathbb{T})$ and then are finite for a.e. $x\in\mathbb{T}$, we conclude $\det{M(x)}\equiv C\in(0,\infty)$. Hence
  \begin{align}\label{eq4.3}
    \overline{A_{\lambda_{2},\lambda_{1},z}(x)}=M(x+\Phi)
    \begin{pmatrix}
    e^{2\pi i\theta} & 0\\
    0 & e^{-2\pi i\theta}
    \end{pmatrix}M(x)^{-1}=M(x+\Phi)
    \begin{pmatrix}
    e^{2\pi i\theta} & 0\\
    0 & e^{-2\pi i\theta}
    \end{pmatrix}\frac{M(x)^{*}}{\det{M}},
  \end{align}
  where $M^{*}$ represents the adjoint matrix of $M$.
  Because the entries $\check{\phi}^{\pm}$ of $M$ are all in $L^{2}(\mathbb{T})$, the right hand side of (\ref{eq4.3}) is in $L^{1}(\mathbb{T})$ by Cauchy-Schwarz inequality.
  However, 
  \[\begin{aligned}\overline{A_{\lambda_{2},\lambda_{1},z}(x)}
=\begin{pmatrix}
  \overline{A_{\lambda_{2},\lambda_{1},z}(x)_{11}} & \overline{A_{\lambda_{2},\lambda_{1},z}(x)_{12}}\\
  \overline{A_{\lambda_{2},\lambda_{1},z}(x)_{21}} & \overline{A_{\lambda_{2},\lambda_{1},z}(x)_{22}}
  \end{pmatrix}\end{aligned}\]
  (where the entries are presented in (\ref{eq.cocycleMap}) with $\lambda_{1}$ and $\lambda_{2}$ exchanged) is not in $L^{1}(\mathbb{T})$ since $\frac{1}{k}(2t\lambda_{1}\cos{2\pi x}+(t^{2}-1))\lambda_{1}')$ has a zero $x_{0}$ on $\mathbb{T}$ and $\frac{1}{k}(2t\lambda_{1}\cos{2\pi x}+(t^{2}-1))\lambda_{1}')$ is an infinitesimal of the same order with $x-x_{0}$ or $(x-x_{0})^{2}$ as $x\to x_{0}$, which is a contradiction. 
\end{proof}

\subsection{Singular continuous spectrum in self-dual regime}
In this section, we will prove purely singular continuous spectrum in the critical case $\{\lambda_{1}=\lambda_{2}\in(0,\frac{\left|1-t^{2}\right|}{1+t^{2}})\}$.

\subsubsection{The absence of absolutely continuous spectrum}

First, we show that any Szeg\H{o} cocycle induced by a quasiperiodic quantum walk is monotonic.
\begin{definition}[\cite{AK2015}]
    Consider an $\SL(2,\mathbb{R})$-cocycle $(\Phi,S_{z})$ and view it as a holomorphic function of $z\in\mathbb{C}\setminus\{0\}$. Then we say the cocycle is monotonic if for any $x\in\mathbb{T}$ and any $v\in S^1$, the map
\[
s\mapsto \arg\bigl(S_{e^{is}}(x)v\bigr)
\]
is monotonic on $(0,2\pi]$.
\end{definition}

We now state this result precisely.
We consider the extended CMV matrix $\mathcal{E}$ with quasiperiodic Verblunsky coefficients $\{\alpha_n,\rho_n\}$ defined by
\[
\alpha_{2n}=\lambda_1', \quad \rho_{2n}=\lambda_1, \qquad 
\alpha_{2n-1}=g(\theta+n\Phi), \quad \rho_{2n-1}=f(\theta+n\Phi),
\]
where $f:\mathbb{T}\to\mathbb{R}_{+}$ and $g:\mathbb{T}\to\mathbb{C}$ are analytic functions. The associated Szeg\H{o} cocycle is given by
\[
S_z^{+}(x)=\frac{1}{\lambda_1 f(x)}
\begin{pmatrix}
z+\lambda_1' g(x) & -\overline{g(x)}-\lambda_1' z^{-1} \\
-g(x)-\lambda_1' z & z^{-1}+\lambda_1'\overline{g(x)}
\end{pmatrix},
\qquad z\in\mathbb{C}\setminus\{0\}.
\]
Recall that
\[
\mathcal{P}:=\frac{1}{\sqrt{2}}
\begin{pmatrix}
1 & i \\
i & 1
\end{pmatrix}
\]
induces an isomorphism between $\SU(1,1)$ and $\SL(2,\mathbb{R})$. Define
$S_z^{\mathcal{R}}(x):=\mathcal{P}S_z^{+}(x)\mathcal{P}^{*},$ then we the basic observation is the following: 

\begin{lemma}\label{lem9.2}
For any $\Phi\in\mathbb{R}\setminus\mathbb{Q}$, the cocycle $(\Phi,S_z^{\mathcal{R}})$ is monotonic. 
\end{lemma}
\begin{proof}
A direct computation yields
\[
S_z^{\mathcal{R}}(x)=\frac{1}{\lambda_1 f(x)}
\begin{pmatrix}
S_{11}(x) & S_{12}(x) \\
S_{21}(x) & S_{22}(x)
\end{pmatrix},
\]
where
\begin{align*}
S_{11}(x)&=\Re z+\lambda_1'\Im z+\lambda_1'\Re g(x)+\Im g(x),\\
S_{12}(x)&=-\lambda_1'\Re z+\Im z-\Re g(x)+\lambda_1'\Im g(x),\\
S_{21}(x)&=-\lambda_1'\Re z-\Im z-\Re g(x)-\lambda_1'\Im g(x),\\
S_{22}(x)&=\Re z-\lambda_1'\Im z+\lambda_1'\Re g(x)-\Im g(x).
\end{align*}

Let $v=(v_1,v_2)^{\top}$ with $v_1^2+v_2^2=1$, and define
\[
\lambda_1 f(x) S_{e^{is}}^{\mathcal{R}}(x)v=(y_1(s),y_2(s))^{\top}.
\]
Then
\[
\frac{\mathrm{d}}{\mathrm{d}s}\arg\bigl(S_{e^{is}}^{\mathcal{R}}(x)v\bigr)
=\frac{y_1(s)y_2'(s)-y_2(s)y_1'(s)}{y_1^2(s)+y_2^2(s)}.
\]

We compute
\begin{align*}
y_1(s)y_2'(s)-y_2(s)y_1'(s)
&=v^{\top}\lambda_1 f(x)(S_{e^{is}}^{\mathcal{R}})^{\top}
\begin{pmatrix}
0 & 1\\
-1 & 0
\end{pmatrix}
\lambda_1 f(x)\frac{\mathrm{d}}{\mathrm{d}s}S_{e^{is}}^{\mathcal{R}}\,v\\
&=:v^{\top}
\begin{pmatrix}
E & F_1\\
F_2 & G_1
\end{pmatrix}v.
\end{align*}
Here  
\[
E=\lambda_1^2\bigl(-\Im g(x)\Re z-\Re g(x)\Im z-1\bigr)
\le \lambda_1^2\bigl(|g(x)|-1\bigr)<0,
\]
and
\begin{align*}
\det
\begin{pmatrix}
E & F_1\\
F_2 & G_1
\end{pmatrix}
&=\lambda_1^2 f^2(x)\det S_{e^{is}}^{\mathcal{R}}
\cdot\det\!\left(\lambda_1 f(x)\frac{\mathrm{d}}{\mathrm{d}s}S_{e^{is}}^{\mathcal{R}}\right)\\
&=\lambda_1^4 f^2(x)>0.
\end{align*}

Moreover, the matrix $\begin{pmatrix}E&F_1\\F_2&G_1\end{pmatrix}$ is symmetric. Indeed,
\[
\frac{\mathrm{d}}{\mathrm{d}s}\!\left((S_{e^{is}}^{\mathcal{R}})^{\top}
\begin{pmatrix}
0 & 1\\
-1 & 0
\end{pmatrix}
S_{e^{is}}^{\mathcal{R}}\right)=0,
\]
since $\det S_{e^{is}}^{\mathcal{R}}=1$. This implies
\[
(S_{e^{is}}^{\mathcal{R}})^{\top}
\begin{pmatrix}
0 & 1\\
-1 & 0
\end{pmatrix}
\frac{\mathrm{d}}{\mathrm{d}s}S_{e^{is}}^{\mathcal{R}}
=
\left[(S_{e^{is}}^{\mathcal{R}})^{\top}
\begin{pmatrix}
0 & 1\\
-1 & 0
\end{pmatrix}
\frac{\mathrm{d}}{\mathrm{d}s}S_{e^{is}}^{\mathcal{R}}\right]^{\top}.
\]

Since the quadratic form is real-symmetric with $E<0$ and positive determinant, it is negative definite. Consequently,
$y_1(s)y_2'(s)-y_2(s)y_1'(s)<0,$
which implies 
\[
\frac{\mathrm{d}}{\mathrm{d}s}\arg\bigl(S_{e^{is}}^{\mathcal{R}}(x)v\bigr)<0
\]
for all $x\in\mathbb{T}$ and all $v\in S^1$. Therefore, the map
$s\mapsto \arg(S_{e^{is}}^{\mathcal{R}}(x)v)$ is monotonically decreasing on $(0,2\pi]$.
\end{proof}

As a consequence of monotonicity of cocycles, we apply the Kotani's theory 
 \cite[Theorem 1.1, 1.7]{AK2015}, and obtain the following: 
\begin{lemma}\label{lem9.3}
    For Lebesgue almost every $z\in\partial\mathbb{D}$ with $L(\Phi,S_{z}^{+})=0$, the cocycle $(\Phi,S_{z}^{\mathcal{R}}(x))$ is $C^{\omega}$-reducible to rotations. That is, there exists $B\in C^{\omega}(\mathbb{T},\mathrm{SL}(2,\mathbb{R}))$ such that 
    $$B^{-1}(x+\Phi)S_{e^{is}}^{\mathcal{R}}(x)B(x)\in\mathrm{SO}(2,\mathbb{R}).$$
\end{lemma}

Within these preparation, we obtain 
\begin{lemma}\label{lem9.5}
    For any $\lambda_{1}=\lambda_{2}\in(0,\frac{\left|1-t^{2}\right|}{1+t^{2}})$ and $\Phi$ irrational, $\Sigma$  has zero Lebesgue measure. 
\end{lemma}
\begin{proof}
    By Theorem \ref{thm6.4}, when $\lambda_{1}=\lambda_{2}\in(0,\frac{\left|1-t^{2}\right|}{1+t^{2}})$, and for any $z\in \Sigma$, the cocycle $(\Phi,S_{z}^{+})$ is critical.  Suppose that $\Sigma$ has positive Lebesgue measure, then for almost every $z\in\Sigma$, $(\Phi,S_{z}^{+})$ are rotations reducible, which contradicts to the fact that $(\Phi,S_{z}^{+})$ is critical. 
\end{proof}

\subsubsection{Absence of  point spectrum}
In this section we prove the absence of point spectrum in the critical case. Different with region $\{\lambda_{1},\lambda_{2}\in[\frac{\left|1-t^{2}\right|}{1+t^{2}},1)\}$, the cocycle map $A_{z}(x)=A_{\lambda_1,\lambda_2,z}(x)$ has no singularities in the case $\lambda_1=\lambda_2$, leading to a more involved argument.
\begin{lemma}\label{lem9.6}
    For $\lambda_{1}=\lambda_{2}\in(0,\frac{\left|1-t^{2}\right|}{1+t^{2}})$, $W_{\lambda_{1},\lambda_{1},\theta,\Phi}$ has no eigenfunction for any irrational $\Phi$ and $\theta$ irrational with respect to $\Phi$.
\end{lemma}
We will use the ideas of \cite{AJM2017} in the proof.
\begin{proof}
    When $\lambda_{1}=\lambda_{2}\in(0,\frac{\left|1-t^{2}\right|}{1+t^{2}})$, assume that $W_{\lambda_{1},\lambda_{1},\theta,\Phi}$ has an eigenvalue $z\in\partial\mathbb{D}$ and an eigenfunction $\psi\in\mathcal{H}$. Then by the proof of Lemma \ref{prop.noEigen2}, 
    $$\overline{A_{\lambda_{1},\lambda_{1},z}(x)}=M(x+\Phi)
    \begin{pmatrix}
    e^{2\pi i\theta} & 0\\
    0 & e^{-2\pi i\theta}
    \end{pmatrix}M(x)^{-1},$$
    where $M(x)=\begin{pmatrix}
  \check{\phi}^{+}(x) & \check{\phi}^{-}(-x)\\
  e^{-2\pi i\theta}\check{\phi}^{-}(x-\Phi) & -e^{2\pi i\theta}\check{\phi}^{+}(-x+\Phi)
  \end{pmatrix}$ and
  $$\overline{A_{\lambda_{1},\lambda_{1},z}(x)}=
  \begin{pmatrix}
 \overline{A_{\lambda_{1},\lambda_{1},z}(x)}_{11} & \overline{A_{\lambda_{1},\lambda_{1},z}(x)}_{12} \\
  \overline{A_{\lambda_{1},\lambda_{1},z}(x)}_{21} & \overline{A_{\lambda_{1},\lambda_{1},z}(x)}_{22}
  \end{pmatrix}$$
  with entries specified as the following: 
$$\overline{A_{\lambda_{1},\lambda_{1},z}(x)}_{11}=\frac{1}{f_{\lambda_{1}}(x)}[\lambda_{1}^{-1}z+\frac{1}{k}\lambda_{1}'\lambda_{1}^{-1}((1-t^{2})\lambda_{1}(e^{2\pi ix}+e^{-2\pi ix})+4t\lambda_{1}')+z^{-1}\lambda_{1}'^{2}\lambda_{1}^{-1}],$$
$$\overline{A_{\lambda_{1},\lambda_{1},z}(x)}_{12}=\frac{1}{f_{\lambda_{1}}(x)}[-\frac{1}{k}(-t^{2}\lambda_{1}e^{-2\pi ix}+\lambda_{1}e^{2\pi ix}+2t\lambda_{1}')-\lambda_{1}'z^{-1}],$$
$$\overline{A_{\lambda_{1},\lambda_{1},z}(x)}_{21}=\frac{1}{f_{\lambda_{1}}(x)}[-\frac{1}{k}(-t^{2}\lambda_{1}e^{2\pi ix}+\lambda_{1}e^{-2\pi ix}+2t\lambda_{1}')-\lambda_{1}'z^{-1}],$$
\[\overline{A_{\lambda_{1},\lambda_{1},z}(x)}_{22}=\frac{\lambda_{1}z^{-1}}{f_{\lambda_{1}}(x)},\]
where $f_{\lambda_{1}}(x)=\frac{1}{k}\bigl(2t\lambda_{1}\cos(2\pi x)+(t^{2}-1)\lambda_{1}'\bigr).$

Hence for $n\ge1$, the $n$-step transfer matrix satisfies 
$$(\overline{A_{\lambda_{1},\lambda_{1},z}})^{n}(x)=\prod_{j=n-1}^{0}\overline{A_{\lambda_{1},\lambda_{1},z}}(x+j\Phi)=M(x+n\Phi)
    \begin{pmatrix}
    e^{2\pi in\theta} & 0\\
    0 & e^{-2\pi in\theta}
    \end{pmatrix}M(x)^{-1}.$$
We define 
$$\Psi_{n}(x):=\tr \left[(\overline{A_{\lambda_{1},\lambda_{1},z}})^n(x)\right]-\tr\left[\begin{pmatrix}
    e^{2\pi in\theta} & 0\\
    0 & e^{-2\pi in\theta}
    \end{pmatrix}\right]=\tr \left[(\overline{A_{\lambda_{1},\lambda_{1},z}})^n(x)\right]-(e^{2\pi in\theta}+e^{-2\pi in\theta}).$$
Moreover, 
$$\Psi_{n}(x)=\tr[(M(x+n\Phi)-M(x))
    \begin{pmatrix}
    e^{2\pi in\theta} & 0\\
    0 & e^{-2\pi in\theta}
    \end{pmatrix}M(x)^{-1}].$$
We have
\begin{align*}
    \left|\Psi_{n}(x)\right|&\le2\left\|[M(x+n\Phi)-M(x)]
    \begin{pmatrix}
    e^{2\pi in\theta} & 0\\
    0 & e^{-2\pi in\theta}
    \end{pmatrix}M(x)^{-1}\right\|\\
    &\le2\left\|[M(x+n\Phi)-M(x)]\right\|\frac{\left\|M(x)\right\|}{\left|\det{M(x)}\right|}.
\end{align*}
By Cauchy-Shwarz inequality, 
$$\left\|\Psi_{n}(x)\right\|_{L^{1}}\le2\left\|[M(x+n\Phi)-M(x)]\right\|_{L^{2}}\frac{\left\|M(x)\right\|_{L^{2}}}{\left|\det{M(x)}\right|}$$
since $\left\|M(x)\right\|_{L^{2}}<\infty$. Because $n\Phi$ tends to $0$ infinitely many times on $\mathbb{T}$, we have 
$$\liminf_{n
\to\infty}\left\|\Psi_{n}(x)\right\|_{L^{1}}=0$$
by continuity of the $L^{2}$-norm with respect to translation.

We can directly calculate that $\Psi_{n}(x)=\frac{\Psi_{n}'(x)}{T_{n}(x)}$ where $T_{n}(x)=\prod_{j=0}^{n-1}f(x+j\Phi)$ and $\Psi_{n}'(x)$ is a trigonometric polynomial of degree at most $n$. Let
$$\tilde{\gamma}=\int_{\mathbb{T}}{\ln{\left|f(\theta)\right|\lambda_{1}}d\theta}=
\ln{\frac{\lambda_{1}}{2(1+t^{2})}(\left|1-t^{2}\right|\lambda_{2}'+\sqrt{(1+t^{2})^{2}\lambda_{2}'^{2}-4t^{2}})},$$
then by \cite[Theorem 2.3]{AJM2017}, we have
$$\liminf_{n
\to\infty}\left\|\lambda_{1}^{n}e^{-\tilde{\gamma}n}\Psi_{n}'(x)\right\|_{L^{1}}=0.$$
Let $\Psi_{n}'^{(n)}$ be the coefficient of $e^{2\pi inx}$ of the trigonometric polynomial $\Psi_{n}'(x)$. 

Denote $\mathcal{Q}=\begin{pmatrix}
    \frac{1-t^{2}}{1+t^{2}}\lambda_{1}' & -\frac{1}{1+t^{2}}\lambda_{1} \\
    \frac{t^{2}}{1+t^{2}}\lambda_{1} & 0
\end{pmatrix},$ then we have
$$T_{n}(x)\tr (\overline{A_{\lambda_{1},\lambda_{1},z}})^{n}(x)=\tr\prod_{j=0}^{n-1}\left(-
\mathcal{Q}e^{2\pi in(x+j\Phi)}+Be^{-2\pi in(x+j\Phi)}+C\right),$$
where $B,C$ are two constant matrices depending only on $\lambda_{1}$, and the eigenvalues of $\mathcal{Q}$  
are 
$$e_{\pm}=\frac{1}{1+t^{2}}\frac{(1-t^{2})\lambda_{1}'\pm\sqrt{(1+t^{2})^{2}\lambda_{1}'^{2}-4t^{2}}}{2}.$$
Hence
\begin{align*}
    \left|\Psi_{n}'^{(n)}\right|&\ge\left(\frac{\left|1-t^{2}\right|\lambda_{1}'+\sqrt{(1+t^{2})^{2}\lambda_{1}'^{2}-4t^{2}}}{2(1+t^{2})}\right)^{n}+\left(\frac{\left|1-t^{2}\right|\lambda_{1}'-\sqrt{(1+t^{2})^{2}\lambda_{1}'^{2}-4t^{2}}}{2(1+t^{2})}\right)^{n}\\
    &-2\left(\frac{\left|t\right|}{1+t^{2}}\lambda_{1}\right)^{n},
\end{align*}
\begin{align*}
    \lambda_{1}^{n}e^{-\tilde{\gamma}n}\left|\Psi_{n}'^{(n)}\right|&\ge 1+\left(\frac{\left|1-t^{2}\right|\lambda_{1}'-\sqrt{(1+t^{2})^{2}\lambda_{1}'^{2}-4t^{2}}}{\left|1-t^{2}\right|\lambda_{1}'+\sqrt{(1+t^{2})^{2}\lambda_{1}'^{2}-4t^{2}}}\right)^{n}
    -2\left(\frac{2\left|t\right|\lambda_{1}}{\left|1-t^{2}\right|\lambda_{1}'+\sqrt{(1+t^{2})^{2}\lambda_{1}'^{2}-4t^{2}}}\right)^{n}\\
    &\to 1 \text{ as } n\to\infty. 
\end{align*}
However,
$$\lambda_{1}^{n}e^{-\tilde{\gamma}n}\left|\Psi_{n}'^{(n)}\right|=\left|\int_{\mathbb{T}}\lambda_{1}^{n}e^{-\tilde{\gamma}n}\Psi_{n}'(x)e^{-2\pi inx}dx\right|\le\left\|\lambda_{1}^{n}e^{-\tilde{\gamma}n}\Psi_{n}'(x)\right\|_{L^{1}},$$
therefore 
\[\liminf_{n
\to\infty}\lambda_{1}^{n}e^{-\tilde{\gamma}n}\left|\Psi_{n}'^{(n)}\right|=0,\]
this is a contradiction.
\end{proof}

\begin{proof}[Proof of Theorem \ref{thm.main}(2) and Theorem \ref{main2}(1)]
Direct consequence of Lemma \ref{lem9.5} and Lemma \ref{lem9.6}. 
\end{proof}

\section{Localization in the positive Lyapunov exponent region}
In this section, we prove Part \ref{ItemLocalization} of Theorem \ref{thm.main} and Part \ref{ItemLocalization2} of Theorem \ref{main2}. The proof will proceed in the similar way as in Jitomirskaya's seminal paper \cite{Jitomirskaya1999} for the AMO. A crucial ingredient in the proof is to exploit the evenness of the characteristic polynomial of the finite cut-off  of the  matrix representing the quantum walk operator $W_{\lambda_1,\lambda_2,\Phi,\theta}$. Another important ingredient in the proof is to conjugate this matrix into an extended CMV matrix so that the transfer matrix of its generalized eigenvalue equation lies in the group $\mathrm{SU}(1,1)$ instead of $\mathrm{U}(2)$. This will allow us to utilize some useful analysis for the well studied normalized Szeg\H{o} cocycles.
Let us first set the stage for the proof.

\begin{definition}\label{PhiResonant}
    For $\Phi\in DC$ satisfying the condition \eqref{eq.dio} with a constant $1<\tau=\tau(\Phi)<\infty$.
    We say $\theta\in\mathbb{T}$ is  {\it resonant} with respect to $\Phi$ if 
    \[\left|\sin2\pi(\theta+n\Phi)\right|<e^{-\left|n\right|^{\frac{1}{2\tau}}}\]
    holds for infinitely many $n$ satisfying $2n\in\mathbb{Z}$. Otherwise, $\theta$ is called {\it non-resonant} with respect to $\Phi$. 
\end{definition}
 Let $\Theta=\Theta(\Phi)$ be the set of resonant phases. It is known that $\Theta$ is a zero measure dense $G_\delta$ set. Therefore, the complement $\Theta^c$ is of full measure. The main goal of this section is to prove the following result:

\begin{theorem}\label{lem10.2}
    Suppose that \(\Phi\in DC, \theta\notin\Theta,  \lambda_1<\frac{|1-t^2|}{1+t^2}, \lambda_1<\lambda_2\). Then the quantum walk operator $W_{\lambda_1,\lambda_2,\Phi,\theta}$ has purely point spectrum with exponentially decaying eigenvectors.
\end{theorem}

Let $\mathcal{E}$ be the matrix representation of the quantum walk operator $W_{\lambda_1,\lambda_2,\Phi,\theta}$. In our case, $\mathcal{E}$ is a GECMV since the way we generate $\rho_n$'s can not guarantee the positivity of each $\rho_n$'s. Nevertheless, we can apply Lemma \ref{lem.gaugeTransform} to find a unitary matrix $U$ such that $\tilde{\mathcal{E}}=U\mathcal{E}U^*$ is an extended CMV matrix whose associated $\rho_n$'s are all real positive. In the following context, we will mostly work with $\tilde{\mathcal{E}}$ instead of $\mathcal{E}$. We remind the reader that the difference between $\tilde{\mathcal{E}}$ with $\mathcal{E}$ is simply that $\tilde{\mathcal{E}}=\mathcal{E}(\alpha_n,|\rho_n|)$ while $\mathcal{E}=\mathcal{E}(\alpha_n,\rho_n)$ where $\mathcal{E}=\mathcal{L}\mathcal{M}$ in the factorization with building blocks $\Theta_n$ in \eqref{GECMV}.

According to \cite{Simon82}, it sufficies to prove that  any nontrivial $\psi\in\mathcal{H}$ satisfying $\tilde{\mathcal{E}}\psi=z\psi$ and $\Vert\psi_n\Vert_{\bbC^2}\leq M(1+|n|)^{N}$ for some $M,N>0$ decays exponentially for every $z\in \Sigma$ satisfying $L(z)>0.$


Let us first introduce the Green's function.
Let $\Lambda=[a,b]\cap\mathbb{Z}$ be a finite interval and $\beta_{1},\beta_{2}\in\partial\mathbb{D}\cup\{\cdot\}$. Given $\{\alpha_{n}\}\subset\mathbb{D}$, define $\{\tilde{\alpha}_{n}\}$ as follows:
$$\tilde{\alpha}_{j}:=
\begin{cases}
    \beta_{1}, j=a-1,\\
    \alpha_{j}, j\neq a-1,b,\\
    \beta_{2}, j=b.
\end{cases}$$
Let $\tilde{\mathcal{E}}^{\beta_{1},\beta_{2}}$ be the CMV matrix with Verblunsky coefficients $\{\tilde{\alpha}_{n},\left|\rho_{n}\right|\}$ and define $\tilde{\mathcal{E}}_{\Lambda}^{\beta_{1},\beta_{2}}:=\chi_{\Lambda}\tilde{\mathcal{E}}^{\beta_{1},\beta_{2}}\chi_{\Lambda}^{*}$, where $\chi_{\Lambda}$ is the finite projection onto $\Lambda$. Then $\tilde{\mathcal{E}}_{\Lambda}^{\beta_{1},\beta_{2}}$ is unitary when $\beta_{1},\beta_{2}\in\partial\mathbb{D}$. In particular, let $\tilde{\mathcal{E}}_{\Lambda}^{\cdot,\beta_{2}}:=\tilde{\mathcal{E}}_{\Lambda}^{\alpha_{a-1},\beta_{2}}$ and $\tilde{\mathcal{E}}_{\Lambda}^{\beta_{1},\cdot}:=\tilde{\mathcal{E}}_{\Lambda}^{\beta_{1},\alpha_{b}}$.
Define $\rho_{\Lambda}=\prod_{j\in\Lambda}\rho_{j}$ and $P_{z,\Lambda}^{\beta_{1},\beta_{2}}=\det{(zI-\tilde{\mathcal{E}}_{\Lambda}^{\beta_{1},\beta_{2}})}$. For $a>b$, we take $P_{z,\Lambda}^{\beta_{1},\beta_{2}}=1$.

Consider the eigenvalue equation $(zI-\tilde{\mathcal{E}})\psi^{z}=(zI-\tilde{\mathcal{L}}\tilde{\mathcal{M}})\psi^{z}=0$, i.e.  $(z\tilde{\mathcal{L}}^{*}-\tilde{\mathcal{M}})\psi^{z}=0$. Define the finite-volume Green's function as 
$$G_{z,\Lambda}^{\beta_{1},\beta_{2}}:=(z(\tilde{\mathcal{L}}_{\Lambda}^{\beta_{1},\beta_{2}})^{*}-\tilde{\mathcal{M}}_{\Lambda}^{\beta_{1},\beta_{2}})^{-1}$$
and $G_{z,\Lambda}^{\beta_{1},\beta_{2}}(x,y)=\left \langle \delta_{x}, G_{z,\Lambda}^{\beta_{1},\beta_{2}}\delta_{y}\right \rangle$ for $x,y\in\Lambda$. The eigenfunction evaluated at an internal site $a<y<b$ can be represented by the boundary conditions and the Green's function as the following (c. f. \cite[Lemma 3.1]{Zhu2024})
\begin{equation}\label{eq.boundary2Interior}\psi^{z}(y)=G_{z,\Lambda}^{\beta_{1},\beta_{2}}(y,a)\tilde{\psi}^{z}(a)+G_{z,\Lambda}^{\beta_{1},\beta_{2}}(y,b)\tilde{\psi}^{z}(b),
\end{equation}
where
$$\tilde{\psi}^{z}(a)=
\begin{cases}
    \rho_{a-1}\psi^{z}(a-1)-(\alpha_{a-1}-\beta_{1})\psi^{z}(a), a\,is\,even,\\
    (z\overline{\alpha_{a-1}}-z\overline{\beta_{1}})\psi^{z}(a)-z\rho_{a-1}\psi^{z}(a-1), a\,is\,odd
\end{cases}$$
and
$$\tilde{\psi}^{z}(b)=
\begin{cases}
    (z\beta_{2}-z\alpha_{b})\psi^{z}(b)-z\rho_{b}\psi^{z}(b+1), b\,is\,even,\\
    (\overline{\alpha_{b}}-\overline{\beta_{2}})\psi^{z}(b)+\rho_{b}\psi^{z}(b+1), b\,is\,odd.
\end{cases}$$
By Cramer's rule, 
\begin{equation}\label{eq.GreenRep}\left|G_{z,\Lambda}^{\beta_{1},\beta_{2}}(x,y)\right|=\left|\rho_{[x,y-1]}\frac{P_{z,[a,x-1]}^{\beta_{1},\cdot}P_{z,[y+1,b]}^{\cdot,\beta_{2}}}{P_{z,\Lambda}^{\beta_{1},\beta_{2}}}\right|, \,x,y\in\Lambda.
\end{equation}

\begin{definition}
    Fix $z=e^{it}\in\partial\mathbb{D}$, $\gamma\in\mathbb{R}$ and $k\in\mathbb{Z}$. We say that $y\in\mathbb{Z}$ is {\it $(\gamma,k)$-regular} if
\begin{enumerate}
    \item there exists $[n_{1},n_{2}]$ containing $y$ such that $n_{2}-n_{1}+1=k$,
    \item $\left|y-n_{i}\right|\ge\frac{k}{7},i=1,2$,
    \item $\left|G_{z,[n_{1},n_{2}]}^{\beta_{1},\beta_{2}}(y,n_{i})\right|<e^{-\gamma\left|y-n_{i}\right|},i=1,2$.
\end{enumerate}
Otherwise, we say that $y$ is {\it $(\gamma,k)$-singular}.
\end{definition}

The following result is the main technical lemma of this section.
\begin{lemma}\label{lem10.7}
    Let $\Phi\in DC$, $\theta\notin\Theta$ and assume that $L(z)>0$. Then for any $\epsilon>0$, there exists \(k_0=k_{0}(\theta,\Phi,z,\epsilon)\in\mathbb{N}\), s.t. for any $\left|y\right|>k_{0}$, there exists $k>\frac{5}{16}\left|y\right|$, such that $y$ is $(\frac{L(z)}{2}-\epsilon,k)$-regular.    
\end{lemma}

Let us assume that Lemma \ref{lem10.7} holds and proceed to the proof Theorem \ref{lem10.2}. We will postpone the proof of Lemma \ref{lem10.7} to the end of this section.

\begin{proof}[Proof of Theorem \ref{lem10.2}]
     From the expression of the Lyapunov exponents specified by Theorem \ref{thm6.4}, Theorem \ref{thm6.11} and Lemma \ref{cor.simpleLE}, we know $L(z)>0$ for $z\in\Sigma$ and $\lambda_1<\frac{|1-t^2|}{1+t^2}, \lambda_1<\lambda_2$. Let $\psi_{z}$ be a nontrivial generalized eigenfunction of $\tilde{\mathcal{E}}(\theta)$. Take $\gamma=\frac{L(z)}{2}$ and let $\epsilon>0$ be sufficiently small. For such $\epsilon>0$, by Lemma \ref{lem10.7} there exists a $k_0>0$ depending on $\theta,\Phi,z,\epsilon$ such that for any  $\left|y\right|>k_{0}$, there exists $k>\frac{5}{16}|y|$ such that $y$ is $(\gamma-\epsilon,k)$-regular. Therefore, by the definition of regularity, 
    \[\left|G_{z,[n_{1},n_{2}]}^{\beta_{1},\beta_{2}}(y,n_{i})\right|<e^{-(\gamma-\epsilon)\left|y-n_{i}\right|}\le e^{-(\gamma-\epsilon)\frac{1}{7}\times\frac{5}{16}\left|y\right|}.\]
    
    Plugging this estimate into \eqref{eq.boundary2Interior}, we obtain 
    $$\left|\psi_{z}(y)\right|\le e^{-(\gamma-\epsilon)\left|y-n_{1}\right|}\left|\tilde{\psi}_{z}(n_{1})\right|+e^{-(\gamma-\epsilon)\left|y-n_{2}\right|}\left|\tilde{\psi}_{z}(n_{2})\right|.$$
    Combining with the trivial bounds 
    $$\left|\tilde{\psi}_{z}(n_{1})\right|\le3\max\{\left|\psi_{z}(n_{1})\right|,\left|\psi_{z}(n_{1}-1)\right|\},$$
    $$\left|\tilde{\psi}_{z}(n_{2})\right|\le3\max\{\left|\psi_{z}(n_{2})\right|,\left|\psi_{z}(n_{2}+1)\right|\},$$
    gives rise to the following estimates:
    \[\left|\psi_{z}(y)\right|\le e^{-(\gamma-\epsilon)\left|y-n_{1}\right|}3\max\{\left|\psi_{z}(n_{1})\right|,\left|\psi_{z}(n_{1}-1)\right|\}+e^{-(\gamma-\epsilon)\left|y-n_{2}\right|}3\max\{\left|\psi_{z}(n_{2})\right|,\left|\psi_{z}(n_{2}+1)\right|\}.\]
    Since $\left|\psi_{z}(y)\right|\le M(1+\left|y\right|)^{N}$ for any $y$, 
    $$\left|\psi_{z}(y)\right|\le 3M(e^{-(\gamma-\epsilon)\left|y-n_{1}\right|}(2+\left|n_{1}\right|)^{N}+e^{-(\gamma-\epsilon)\left|y-n_{2}\right|}(2+\left|n_{2}\right|)^{N}).$$
    For $i=1,2,$\[(2+\left|n_{i}\right|)^{N}\le (2+\left|n_{i}-y\right|+\left|y\right|)^{N}\le(2+2\max\{\left|n_{i}-y\right|,\left|y\right|\})^{N},\]
    we have
    $$\left|\psi_{z}(y)\right|\le e^{-(\gamma-\epsilon)\frac{1}{14}\times\frac{5}{16}\left|y\right|}$$
    for $\left|y\right|>k_{0}$ large enough.

    Therefore $\tilde{\mathcal{E}}$ has purely point spectrum with exponentially decaying eigenvectors.  By Lemma \ref{lem.gaugeTransform}, the same conclusion holds for $\mathcal{E}$ as well. 
\end{proof}

It is readily seen that Part \ref{ItemLocalization} of Theorem \ref{thm.main} and Part \ref{ItemLocalization2} of Theorem \ref{main2} follows from Theorem \ref{lem10.2}.

\subsection{Proof of Lemma \ref{lem10.7}}
In order to prove Lemma \ref{lem10.7}, we need to show that the numerators appeared in \eqref{eq.GreenRep} are properly bounded above and the denominator appeared is properly bounded from below. The key idea is to connect the characteristic polynomials to the behavior of the associated Szeg\H{o} cocycles.

Let $S_z(n)$ be the normalized Szeg\H{o} cocycle such that $S_z(n)\in \mathrm{SU}(1,1)$ and denote
$$\tilde{S}_{z}(n)=\sqrt{z}\left|\rho_{n}\right|S_{z}(n)=
    \begin{pmatrix}
        z & -\overline{\alpha_{n}} \\
        -\alpha_{n}z & 1
    \end{pmatrix}.$$
The following relation can be found in \cite{Zhu2024}: 
\begin{align}\label{eq10.1}
    \begin{pmatrix}
    P_{z,[a,b]}^{\beta_{1},\beta_{2}} & P_{z,[a,b]}^{-\beta_{1},\beta_{2}}\\
    P_{z,[a,b]}^{\beta_{1},-\beta_{2}} & P_{z,[a,b]}^{-\beta_{1},-\beta_{2}}
\end{pmatrix}=
\begin{pmatrix}
    z & -\overline{\beta_{2}}\\
    z & \overline{\beta_{2}}
\end{pmatrix}
\prod_{j=b-1}^{a}\tilde{S}_{z}(j)
\begin{pmatrix}
    1 & 1\\
    -\beta_{1} & \beta_{1}
\end{pmatrix}.
\end{align}
We conclude from it that
$$\left|P_{z,[a,x-1]}^{\beta_{1},\cdot}\right|\le\left\|\prod_{j=x-2}^{a}\tilde{S}_{z}(j)\right\|,\left|P_{z,[y+1,b]}^{\cdot,\beta_{2}}\right|\le\left\|\prod_{j=b-1}^{y+1}\tilde{S}_{z}(j)\right\|.$$
As a consequence of unique ergodicity of the linear map $\theta\mapsto \theta+\Phi$ on $\T$, the following result can be obtained.
\begin{lemma}(\cite{Jitomirskaya1999})\label{lem10.5}
    When $\Phi\in DC$, for any $\eta>0$, there exists $N>0$ such that for any $n>N$,
    $$\prod_{j=0}^{n-1}\left|\rho_{j}\right|\le e^{n(\tilde{\gamma}/2+\eta)}$$
    where $\tilde{\gamma}=\int_{\mathbb{T}}{\ln{\left|f(\theta)\right|\lambda_{1}}d\theta}$.
\end{lemma}
Indeed,  by Lemma \ref{lem5.3.1}, 
\begin{equation}\label{eq.gammatilde}\tilde{\gamma}=
\begin{cases}
    \ln{\frac{\lambda_{1}}{2(1+t^{2})}(\left|1-t^{2}\right|\lambda_{2}'+\sqrt{(1+t^{2})^{2}\lambda_{2}'^{2}-4t^{2}})}, \lambda_{2}\in(0,\frac{|1-t^{2}|}{1+t^{2}}),\\
    \ln{\frac{\left|t\right|\lambda_{1}\lambda_{2}}{1+t^{2}}}, \lambda_{2}\in[\frac{|1-t^{2}|}{1+t^{2}},1).
\end{cases}\end{equation}

By the definition of Lyapunov exponent and Furman's subadditive ergodic theorem,  we have 
\begin{lemma}\label{lem10.8}
    For any $\epsilon>0$, $z\in\partial\mathbb{D}$, there exists $k_{1}(\epsilon,z)>0$, such that 
    $$\left|P_{z,[a,b]}^{\beta_{1},\cdot}\right|,\left|P_{z,[a,b]}^{\cdot,\beta_{2}}\right|<e^{(\gamma+\tilde{\gamma}/2+\epsilon)(b-a)}$$
    whenever $b-a>k_{1}$.
\end{lemma}

We will also need the lower bound of the denominator \(P_{z,\Lambda}^{\beta_1,\beta_2}\) in suitable sense.
\begin{lemma}\label{lem10.9}
    For any $\epsilon>0$, $z\in\partial\mathbb{D}$, there exists $k_{2}(\epsilon,z)>0$ such that 
    $$\frac{1}{2n-2}\int_{\mathbb{T}}{\ln{\left|P_{z,[1,2n-2]}^{\beta_{1},\beta_{2}}(\theta)\right|}d\theta}\ge\gamma+\frac{\tilde{\gamma}}{2}-\epsilon$$
    whenever $2n-2>k_{2}$.
\end{lemma}
\begin{proof} Let $\omega=e^{2\pi i\theta}$, then 
 \begin{align*}
      \tilde{S}_{z}^{+}(\theta)&=
    \begin{pmatrix}
   z^{2}-\lambda_{1}'z\frac{1}{k}(t^{2}\lambda_{2}e^{2\pi i\theta}-\lambda_{2}e^{-2\pi i\theta}-2t\lambda_{2}') & \frac{1}{k}(t^{2}\lambda_{2}e^{-2\pi i\theta}-\lambda_{2}e^{2\pi i\theta}-2t\lambda_{2}')z-\lambda_{1}' \\
   -\lambda_{1}'z^{2}+\frac{1}{k}(t^{2}\lambda_{2}e^{2\pi i\theta}-\lambda_{2}e^{-2\pi i\theta}-2t\lambda_{2}')z & -\lambda_{1}'z\frac{1}{k}(t^{2}\lambda_{2}e^{-2\pi i\theta}-\lambda_{2}e^{2\pi i\theta}-2t\lambda_{2}')+1
\end{pmatrix}\\
&=\omega^{-1}
\begin{pmatrix}
   z^{2}\omega-\lambda_{1}'z\frac{1}{k}(t^{2}\lambda_{2}\omega^{2}-\lambda_{2}-2t\lambda_{2}'\omega) & \frac{1}{k}(t^{2}\lambda_{2}-\lambda_{2}\omega^{2}-2t\lambda_{2}'\omega)z-\lambda_{1}'\omega \\
   -\lambda_{1}'z^{2}\omega+\frac{1}{k}(t^{2}\lambda_{2}\omega^{2}-\lambda_{2}-2t\lambda_{2}'\omega)z & -\lambda_{1}'z\frac{1}{k}(t^{2}\lambda_{2}-\lambda_{2}\omega^{2}-2t\lambda_{2}'\omega)+\omega
\end{pmatrix}\\
&:=\omega^{-1}\hat{S}(\omega),
 \end{align*}

Let
$$V=
\begin{pmatrix}
    1 & 0\\
    0 & 0
\end{pmatrix},
B_{1}=
\begin{pmatrix}
    1 & 1\\
    -\beta_{1} & \beta_{1}
\end{pmatrix},
B_{2}=
\begin{pmatrix}
    z & -\overline{\beta_{2}}\\
    z & \overline{\beta_{2}}
\end{pmatrix},$$
then
\begin{align*}
    \int_{\mathbb{T}}{\ln{\left|P_{z,[1,2n-2]}^{\beta_{1},\beta_{2}}(\theta)\right|}d\theta}&=\int_{\mathbb{T}}{\ln{\left\|VB_{2}\tilde{S}_{z}^{-1}(2n-2)\prod_{j=n-1}^{1}\tilde{S}_{z}^{+}(\theta+j\Phi)B_{1}V\right\|}d\theta}\\
&=\int_{\partial\mathbb{D}}{\ln{\left\|VB_{2}\tilde{S}_{z}^{-1}(2n-2)\prod_{j=n-1}^{1}\omega^{-1}\hat{S}(\omega e^{2\pi ij\Phi})B_{1}V\right\|}d\omega}\\
&\ge\ln{\left\|VB_{2}\tilde{S}_{z}^{-1}(2n-2)\prod_{j=n-1}^{1}\hat{S}(0)B_{1}V\right\|},
\end{align*}
where the last inequality follows from the subharmonicity.

We can calculate the eigenvalues of
\(\hat{S}(0)=
\begin{pmatrix}
   \lambda_{1}'z\frac{1}{k}\lambda_{2} & \frac{1}{k}t^{2}\lambda_{2}z \\
   -\frac{1}{k}\lambda_{2}z & -\lambda_{1}'z\frac{1}{k}t^{2}\lambda_{2}
\end{pmatrix},\) as 
$$\lambda_{\pm}=-z \frac{\lambda_{2}}{2(1+t^{2})}((1-t^{2})\lambda_{1}'\pm\sqrt{(1+t^{2})^{2}\lambda_{1}'^{2}-4t^{2}}).$$
Since $\tilde{S}_{z}(2n-2)^{-1}$ is a constant matrix, we have
$$\lim_{n\to\infty}\frac{1}{2n-2}\ln{\left\|VB_{2}\prod_{j=n-1}^{1}\hat{S}(0)B_{1}V\right\|}=\frac{1}{2}\ln{\max\{\left|\lambda_{\pm}\right|\}}=\gamma+\frac{\tilde{\gamma}}{2}.$$
Hence for any $\epsilon>0$, $z\in\partial\mathbb{D}$, there exists $k_{2}(\epsilon,z)>0$ such that 
    $$\frac{1}{2n-2}\int_{\mathbb{T}}{\ln{\left|P_{z,[1,2n-2]}^{\beta_{1},\beta_{2}}(\theta)\right|}d\theta}\ge\gamma+\frac{\tilde{\gamma}}{2}-\epsilon$$
whenever $2n-2>k_{2}$.
\end{proof}

In Jitomirskaya's scheme \cite{Jitomirskaya1999}, the key insight is that \( P_{z,[1,4n-2]}^{\beta_{1},\beta_{2}}(\theta) \) is a polynomial in \( \cos(\theta) \) (Lemma \ref{lem10.10}). To prove this, the critical observation is the reflection symmetry of the ECMV operator \( \tilde{\mathcal{E}} \), which we now explore.

Consider the quasiperiodic ECMV matrix $\tilde{\mathcal{E}}$ associated to the  coefficients \(\{(\alpha_{n},\left|\rho_{n}\right|)\}_{n\in\bbZ}\). Denote entries of $\tilde{\mathcal{E}}$ by $\tilde{\mathcal{E}}_{ij},i,j\in\mathbb{Z}$. Then we can define the {\it alternating reflection} of $\tilde{\mathcal{E}}$ with respect to the center $-\frac{1}{2}$ by $(\tilde{\mathcal{E}}^{\mathcal{R}})_{ij}:=((-1)^{i+j}\tilde{\mathcal{E}}_{-1-i,-1-j})$. Then for the finite cut-off CMV matrix $\tilde{\mathcal{E}}|_{[-n,n-1]}$, we have
\begin{lemma}\label{lem3.7}
    For any $n\in\mathbb{N}^{*},z\in\mathbb{C}$,
    $$\det{(zI-\tilde{\mathcal{E}}|_{[-n,n-1]})}=\det{(zI-\tilde{\mathcal{E}}^{\mathcal{R}}|_{[-n,n-1]})}.$$
\end{lemma}
\begin{proof}
    We note that $$\tilde{\mathcal{E}}^{\mathcal{R}}|_{[-n,n-1]}=PQ
\tilde{\mathcal{E}}|_{[-n,n-1]}QP=PQ
\tilde{\mathcal{E}}|_{[-n,n-1]}Q^{-1}P^{-1},$$
where 
$$P=\begin{pmatrix}
    1 & & & & \\
       & -1 & & & \\
       & & 1 & & \\
    & & & \ddots & \\
    & & & & -1
\end{pmatrix},
Q=\begin{pmatrix}
    & & & 1 \\
    & & 1 & \\
    & \begin{sideways}$\ddots$\end{sideways} & & \\
    1 & & & 
\end{pmatrix}.$$
So $\tilde{\mathcal{E}}^{\mathcal{R}}|_{[-n,n-1]}$ is similar to $\tilde{\mathcal{E}}|_{[-n,n-1]}$, hence 
$$\det{(zI-\tilde{\mathcal{E}}|_{[-n,n-1]})}=\det{(zI-\tilde{\mathcal{E}}^{\mathcal{R}}|_{[-n,n-1]})}.$$
\end{proof}

The following result is a direct consequence of such symmetry.

\begin{lemma}\label{lem10.10}
    $P_{z,[1,4n-2]}^{\beta_{1},\beta_{2}}(\theta)$ is a polynomial of $\cos{2\pi(\theta+n\Phi)}$ of degree at most $2n-1$.
\end{lemma}
\begin{proof}
    First, we can directly calculate that $P_{z,[1,4n-2]}^{\beta_{1},\beta_{2}}(\theta)$ is a polynomial of $\sin{2\pi\theta},\cos{2\pi\theta}$ of degree at most $2n-1$ by (\ref{eq10.1}). Notice that
$$\alpha_{-n-1}(\theta)=\overline{\alpha_{n-1}(-\theta)},\quad \left|\rho_{-n-1}(\theta)\right|=\left|\rho_{n-1}(-\theta)\right|\in\mathbb{R}_{+}$$
    for every $n$ and then $\tilde{\mathcal{E}}(-\theta)=\tilde{\mathcal{E}}^{\mathcal{R}}(\theta)$.
    Hence 
    $$\det{(zI-\tilde{\mathcal{E}}(-\theta)|_{[-n,n-1]})}=\det{(zI-\tilde{\mathcal{E}}^{\mathcal{R}}(\theta)|_{[-n,n-1]})}=\det{(zI-\tilde{\mathcal{E}}(\theta)|_{[-n,n-1]})}$$
    by Lemma \ref{lem3.7}, which implies that 
 $\det{(zI-\tilde{\mathcal{E}}_{[1-2n,2n-2]}^{\beta_{1},\beta_{2}}(\theta))}$ is an even function if $\beta_{1}=\overline{\beta_{2}}$. Hence $P_{z,[1,4n-2]}^{\beta_{1},\beta_{2}}(\theta-n\Phi)=P_{z,[1-2n,2n-2]}^{\beta_{1},\beta_{2}}(\theta)$ is also an even function of $\theta$ if $\beta_{1}=\overline{\beta_{2}}$. So $P_{z,[1,4n-2]}^{\beta_{1},\beta_{2}}(\theta-n\Phi)$ is a polynomial of $\cos{2\pi\theta}$ of degree at most $2n-1$. Hence $P_{z,[1,4n-2]}^{\beta_{1},\beta_{2}}(\theta)$ is a polynomial of $\cos{2\pi(\theta+n\Phi)}$ of degree at most $2n-1$.
\end{proof}
By Lemma \ref{lem10.10}, there exists a polynomial $Q_{n}$ of degree $2n-1$ such that $P_{z,[1,4n-2]}^{\beta_{1},\beta_{2}}(\theta)=Q_{n}(\cos{2\pi(\theta+n\Phi)})$. For any $n\in\mathbb{N}^{*}$ and $r>0$, define
$$A_{n}^{r}:=\{\theta\in\mathbb{T}:\left|Q_{n}(\cos{2\pi\theta})\right|\le e^{(2n-1)r}\}.$$

Then we have
\begin{lemma}\label{lem10.11}
    If $y$ is $(\gamma-\epsilon,4n-2)$-singular for some $n$ and $\epsilon>0$. Then for any $j\in\mathbb{Z}$ satisfying
    $$y-(2n-1)+n\le j\le y+(n-1)+n,$$
we have $\theta+j\Phi\in A_{n}^{2\gamma+\tilde{\gamma}-\frac{1}{7}\epsilon}$ whenever $4n-2>k_{3}(\gamma,\epsilon)$  is large enough.
\end{lemma}
\begin{proof}
    Since $y$ is $(\gamma-\epsilon,4n-2)$-singular, for any $[n_{1},n_{2}]$ containing $y$ with $n_{2}-n_{1}+1=4n-2$ and $\left|y-n_{i}\right|\ge\frac{4n-2}{7}$, we can assume
    $$\left|G_{z,[n_{1},n_{2}]}^{\beta_{1},\beta_{2}}(y,n_{1})\right|\ge e^{-(\gamma-\epsilon)\left|y-n_{1}\right|}.$$
By (\ref{eq.GreenRep}), 
    $$\left|G_{z,[n_{1},n_{2}]}^{\beta_{1},\beta_{2}}(y,n_{1})\right|=\left|\rho_{[n_{1},y-1]}\frac{P_{z,[y+1,n_{2}]}^{\cdot,\beta_{2}}}{P_{z,[n_{1},n_{2}]}^{\beta_{1},\beta_{2}}}\right|.$$
By Lemma \ref{lem10.8}, for any $\epsilon'>0$, we have $\left|P_{z,[y+1,n_{2}]}^{\cdot,\beta_{2}}\right|<e^{(\gamma+\frac{\tilde{\gamma}}{2}+\epsilon')(n_{2}-y)}$ when $n>k_{1}(\epsilon')$ large enough. By Lemma \ref{lem10.5}, for any $\eta>0$, we have $\left|\rho_{[n_{1},y-1]}\right|\le e^{(y-n_{1})(\tilde{\gamma}/2+\eta)}$ when $n>N(\eta)$ large enough.

    Assume there exists $y-(2n-1)+n\le j\le y+(n-1)+n$ such that $\theta+j\Phi\notin A_{n}^{2\gamma+\tilde{\gamma}-\frac{1}{7}\epsilon}$ for infinitely many $n$. Then $\left|Q_{n}(\cos{2\pi(\theta+j\Phi)})\right|>e^{(2n-1)(2\gamma+\tilde{\gamma}-\frac{1}{7}\epsilon)}$. Hence
$$\left|P_{z,[n_{1},n_{2}]}^{\beta_{1},\beta_{2}}(\theta+(j-n-\frac{n_{1}-1}{2})\Phi)\right|=\left|P_{z,[1,4n-2]}^{\beta_{1},\beta_{2}}(\theta+(j-n)\Phi)\right|>e^{(2n-1)(2\gamma+\tilde{\gamma}-\frac{1}{7}\epsilon)}.$$

    Now put these estimate together, we have
\begin{align*}
    \left|G_{z,[n_{1},n_{2}]}^{\beta_{1},\beta_{2}}(y,n_{1})\right|e^{(\gamma-\epsilon)\left|y-n_{1}\right|}
    &<e^{(y-n_{1})(\tilde{\gamma}/2+\eta)}e^{(\gamma+\frac{\tilde{\gamma}}{2}+\epsilon')(n_{2}-y)}e^{-(2n-1)(2\gamma+\tilde{\gamma}-\frac{1}{7}\epsilon)}e^{(\gamma-\epsilon)\left|y-n_{1}\right|}\\
    &=e^{(y-n_{1})(\gamma+\tilde{\gamma}/2+\eta-\epsilon)}e^{(\gamma+\tilde{\gamma}/2+\epsilon')(n_{2}-y)}e^{-(2n-1)(2\gamma+\tilde{\gamma}-\frac{1}{7}\epsilon)}\\
    &=e^{-(\gamma+\tilde{\gamma}/2)}e^{(y-n_{1})(\eta-\epsilon)}e^{\epsilon'(n_{2}-y)}e^{\frac{1}{7}(2n-1)\epsilon}\\
    &\le e^{-(\gamma+\tilde{\gamma}/2)}e^{\frac{1}{7}(4n-2)(\eta-\epsilon)}e^{\frac{6}{7}\epsilon'(4n-2)}e^{\frac{1}{7}(2n-1)\epsilon}\\
    &=e^{-(\gamma+\tilde{\gamma}/2)}e^{(2n-1)[\frac{2}{7}(\eta-\epsilon)+\frac{12}{7}\epsilon'+\frac{1}{7}\epsilon]}.
\end{align*}
Take $2\eta+12\epsilon'<\epsilon$ and $n>k_{3}(\gamma,\epsilon)$ large enough, we have  $ \left|G_{z,[n_{1},n_{2}]}^{\beta_{1},\beta_{2}}(y,n_{1})\right|e^{(\gamma-\epsilon)\left|y-n_{1}\right|}<1$, i.e. 
$$\left|G_{z,[n_{1},n_{2}]}^{\beta_{1},\beta_{2}}(y,n_{1})\right|<e^{-(\gamma-\epsilon)\left|y-n_{1}\right|},$$
contradiction.
\end{proof}
We write the polynomial $Q_{n}(u)$ of degree $2n-1$ in the form below using Lagrange interpolation.
$$Q_{n}(u)=\sum_{j=0}^{2n-1}{Q_{n}(\cos{2\pi\theta_{j}})\prod_{i\neq j}\frac{u-\cos{2\pi\theta_{i}}}{\cos{2\pi\theta_{j}}-\cos{2\pi\theta_{i}}}}.$$
In order to control the upper bound of $Q_{n}$, we need the following definition.
\begin{definition}
    The sequence $\{\theta_{j}\}_{j=0}^{2n-1}\subset\mathbb{T}$ is called {\it $\epsilon$-uniform} if 
    $$\max_{u\in[-1,1]}\max_{0\le j\le2n-1}\prod_{i=0,i\neq j}^{2n-1}\left|\frac{u-\cos{2\pi\theta_{i}}}{\cos{2\pi\theta_{j}}-\cos{2\pi\theta_{i}}}\right|<e^{(2n-1)\epsilon}.$$
\end{definition}
\begin{lemma}\label{lem10.13}
    Let $0<\epsilon'<\epsilon$, $n\ge1$. If $\{\theta_{j}\}_{j=0}^{2n-1}\subset A_{n}^{2\gamma+\tilde{\gamma}-\epsilon}$, then $\{\theta_{j}\}_{j=0}^{2n-1}$ is not $\epsilon'$-uniform for $4n-2>k_{4}(\epsilon,\epsilon')$ large enough.
\end{lemma}
\begin{proof}
    Since $\{\theta_{j}\}_{j=0}^{2n-1}\subset A_{n}^{2\gamma+\tilde{\gamma}-\epsilon}$, $\left|Q_{n}(\cos{2\pi\theta_{j}})\right|\le e^{(2n-1)(2\gamma+\tilde{\gamma}-\epsilon)}$. Assume $\{\theta_{j}\}_{j=0}^{2n-1}$ is $\epsilon'$-uniform. Then by the interpolation formula, for any $\theta$,
    $$\left|P_{z,[1,4n-2]}^{\beta_{1},\beta_{2}}(\theta)\right|=\left|Q_{n}(\cos{2\pi(\theta+n\Phi)})\right|\le2ne^{(2n-1)(2\gamma+\tilde{\gamma}-\epsilon)}e^{(2n-1)\epsilon'}=2ne^{(2n-1)(2\gamma+\tilde{\gamma}-\epsilon+\epsilon')}.$$
Hence 
$$\frac{1}{4n-2}\int_{\mathbb{T}}{\ln{\left|P_{z,[1,4n-2]}^{\beta_{1},\beta_{2}}(\theta)\right|}d\theta}\le\frac{\ln{2n}}{4n-2}+(\gamma+\tilde{\gamma}/2-(\epsilon-\epsilon')/2).$$
However, by Lemma \ref{lem10.9}, 
$$\frac{1}{4n-2}\int_{\mathbb{T}}{\ln{\left|P_{z,[1,4n-2]}^{\beta_{1},\beta_{2}}(\theta)\right|}d\theta}\ge\gamma+\tilde{\gamma}/2-\epsilon''$$
for any $\epsilon''>0$ and $n>k_{2}(\epsilon'')$ large enough.

We only need take $\epsilon''<\frac{\epsilon-\epsilon'}{2}$ to obtain the contradiction.
\end{proof}
Let $\{\frac{p_{n}}{q_{n}}\}$ be the best rational approximation of $\Phi$, $y$ large enough, $m$ be such that $q_{m}\le\frac{y}{16}<q_{m+1}$ and $s$  the largest positive integer such that $sq_{m}<\frac{y}{16}$. Define
$$I_{1}=[\frac{1+(-1)^{sq_{m}}}{2},sq_{m}]\cap\mathbb{Z}, I_{2}=[1+y-sq_{m},y+sq_{m}]\cap\mathbb{Z}.$$
Then $I_{1}\cap I_{2}=\emptyset$, $\{\theta+j\Phi\}_{j\in I_{1}\cup I_{2}}$ are distinct with each other and $\{\cos{2\pi(\theta+j\Phi)}\}_{j\in I_{1}\cup I_{2}}$ are also distinct since $\theta$ is non-resonant with respect to $\Phi$.

Moreover, by the Appendix A of \cite{CFLOZ24} and \cite{AJ09}, we have the standard conclusion below.
\begin{lemma}\label{lem10.14}
    If $\theta$ is non-resonant with respect to $\Phi$, then for any $\epsilon>0$, $\{\theta+j\Phi\}_{j\in I_{1}\cup I_{2}}$ is $\epsilon$-uniform for $y>y_{0}(\theta,\Phi,\epsilon)$ large enough.
\end{lemma}

\begin{proof}[Proof of Lemma \ref{lem10.7}]
    By Lemma \ref{lem10.14} and Lemma \ref{lem10.13}, $\{\theta+j\Phi\}_{j\in I_{1}\cup I_{2}}$ can not be contained in $A_{(3sq_{m}+(1-(-1)^{sq_{m}})/2)/2}^{2\gamma+\tilde{\gamma}-\epsilon/7}$ for infinitely many $y$. We can suppose $\psi^{z}(0)\neq0$. Then $0$ is $(\gamma-\epsilon,6sq_{m}-(1+(-1)^{sq_{m}}))$-singular for $y$ large enough. Therefore by Lemma \ref{lem10.11}, for any $j\in\mathbb{Z}$ satisfying
    $$-n+1\le j\le2n-1,$$
where $n=(6sq_{m}-(1+(-1)^{sq_{m}})+2)/4$, we have $\theta+j\Phi\in A_{n}^{2\gamma+\tilde{\gamma}-\frac{1}{7}\epsilon}$. Hence $y$ is $(\gamma-\epsilon,6sq_{m}-(1+(-1)^{sq_{m}}))$-regular. In fact, if otherwise, then by Lemma \ref{lem10.11},  for any $j\in\mathbb{Z}$ satisfying
    $$y-n+1\le j\le y+2n-1,$$
where $n=(6sq_{m}-(1+(-1)^{sq_{m}})+2)/4$, we have $\theta+j\Phi\in A_{n}^{2\gamma+\tilde{\gamma}-\frac{1}{7}\epsilon}$. Since $I_{1}\cup I_{2}\subset[-n+1,2n-1]\cup[y-n+1,y+2n-1]$ for $y$ large enough, we have $\{\theta+j\Phi\}_{j\in I_{1}\cup I_{2}}\subset A_{n}^{2\gamma+\tilde{\gamma}-\frac{1}{7}\epsilon}$, which is a contradiction.
Moreover, since $6sq_{m}-(1+(-1)^{sq_{m}})>\frac{5\left|y\right|}{16}$ for $y$ large enough, we have proven Lemma \ref{lem10.7}.
\end{proof}

\section*{Acknowledgements} 
 L. Li was supported by the AMS-Simons Travel Grant 2024--2026.
 Q. Zhou was supported by NSFC grant (12531006,12526201) and  Nankai Zhide Foundation.

\bibliographystyle{siam}
\bibliography{SCUAMO}

{
  \bigskip
  \vskip 0.08in \noindent --------------------------------------

\footnotesize
\medskip
\end{document}